%% file: main.tex
\documentclass[acmsmall]{acmart}

\setcopyright{cc}
\setcctype{by}
\acmJournal{PACMMOD}
\acmYear{2025} \acmVolume{3} \acmNumber{6 (SIGMOD)} \acmArticle{375} \acmMonth{12} \acmPrice{}\acmDOI{10.1145/3769840} 

\input{package}
\AtBeginDocument{%
  }

\begin{document}



\title{Visual Template Inference for Data Extraction from Documents}

\author{Yiming Lin}
\affiliation{%
  \institution{UC Berkeley}
  \city{Berkeley}
  \country{USA}}
\email{yiminglin@berkeley.edu}

\author{Mawil Hasan}
\affiliation{%
  \institution{UC Berkeley}
  \city{Berkeley}
  \country{USA}}
\email{mawil0721@berkeley.edu}

\author{Rohan Kosalge}
\affiliation{%
  \institution{UC Berkeley}
  \city{Berkeley}
  \country{USA}}
\email{rohankosalge@berkeley.edu}

\author{Alvin Cheung}
\affiliation{%
  \institution{UC Berkeley}
  \city{Berkeley}
  \country{USA}}
\email{akcheung@cs.berkeley.edu}

\author{Aditya G. Parameswaran}
\affiliation{%
  \institution{UC Berkeley}
  \city{Berkeley}
  \country{USA}}
\email{adityagp@berkeley.edu}



\input{abstract}



\received{April 2025}
\received[revised]{July 2025}
\received[accepted]{August 2025}
\maketitle

\input{introduction-Aditya-March28}

\input{overview}

\input{template_formulation}

\input{field_prediction}

\input{template_prediction}

\input{extraction}

\input{experiment}
\input{relatedwork}

\input{conclusion}

\input{ack}
\newpage
\bibliographystyle{ACM-Reference-Format}
\bibliography{refs}

\appendixtext{\newpage\input{appendix}}

\end{document}

%% file: package.tex
\usepackage{hyperref}

\usepackage{cleveref}
\usepackage{etoolbox}
\newtoggle{fullversion}
\togglefalse{fullversion}
\usepackage{adjustbox}
\usepackage{caption}
\usepackage{scalerel}
\usepackage{graphicx}
\usepackage{xcolor}
\usepackage{tcolorbox}
\usepackage{amsmath,amsfonts}
\usepackage{ifsym}
\usepackage{mdframed}
\usepackage{listings} 
\usepackage{xparse}
\usepackage[normalem]{ulem}  
\usepackage{enumitem}

\makeatletter
\let\c@lofdepth\relax
\let\c@lotdepth\relax
\makeatother
\usepackage{subfigure}
\tcbset{before={\par\pagebreak[0]\smallskip\parindent=0pt},after={\par\noindent}}
\usepackage{tree-dvips}
\usepackage{multirow}

\usepackage{qtree}

\newcommand{\sys}{\textsc{TWIX}\xspace}

\usepackage[ruled,linesnumbered,vlined]{algorithm2e}
\usepackage{algorithmic}
\usepackage{textcomp}
\usepackage{tikz}
\newcommand*\circled[1]{\tikz[baseline=(char.base)]{
    \node[shape=circle,black,draw,inner sep=0.5pt] (char) {#1};}}
\usetikzlibrary{positioning, shapes.geometric, arrows}
\usetikzlibrary{shapes, arrows.meta, fit, shapes.geometric, fit, backgrounds}
\definecolor{codegreen}{rgb}{0,0.6,0} 
\definecolor{darkblue}{rgb}{0.0, 0.0, 0.5}
\definecolor{darkblueilp}{RGB}{0, 0, 175}
\definecolor{darkred}{rgb}{0.5, 0.0, 0.0}
\definecolor{darkgreen}{rgb}{0.0, 0.5, 0.0}
\definecolor{lightgray}{gray}{0.95}

\definecolor{myForestGreen}{RGB}{34,139,34}

\tcbuselibrary{listings, skins}
\usepackage[disable]{todonotes}
\usepackage[T1]{fontenc}
\usepackage{listings}

\newcommand{\topic}[1]{\vspace{.5pt} \noindent{\bf #1.}}

\newif\ifshowblock

\newcommand{\sigmod}[1]{#1}
\newcommand{\arxiv}[1]{}

\newcommand{\techreport}[1]{}
\newcommand{\techreportsp}[1]{}
\newcommand{\papertext}[1]{#1}
\newcommand{\appendixtext}[1]{#1}

\newcommand{\delete}[1]{\textcolor{red}{}}
\newcommand{\update}[1]{\textcolor{black}{#1}}

\newcommand{\yiming}[1]{\textcolor{black}{#1}}

\newcommand{\revised}[1]{\textcolor{black}{#1}}

\newcommand{\rone}[1]{\textcolor{black}{#1}}
\newcommand{\rtwo}[1]{\textcolor{black}{#1}}
\newcommand{\rthree}[1]{\textcolor{black}{#1}}
\newcommand{\rmix}[1]{\textcolor{black}{#1}}
\newcommand{\anonymous}[2]{#2} 

\newcounter{definition}
\newenvironment{definition}[1][]{\refstepcounter{definition}\par\smallskip\textsc{Definition~\thedefinition.\ #1}}{\smallskip}

\newcounter{assumption}

\newcommand{\code}[1]{\texttt{\small #1}}

\newcounter{example}
\newenvironment{example}[1][]{\refstepcounter{example}\par\smallskip\textsc{Example~\theexample.\ #1}}{$\square$\smallskip}

\newcounter{theorem}
\newenvironment{theorem}[1][]{\refstepcounter{theorem}\par\smallskip\textsc{Theorem~\thetheorem.\ #1}}{\smallskip}

\newcounter{proposition}
\newenvironment{proposition}[1][]{
    \refstepcounter{proposition}
    \par\smallskip
    \textsc{Proposition~\theproposition.} \textit{#1}
}{
    \smallskip
}

\newcommand{\squishlist}{
	\begin{list}{$\bullet$}
		{
			\setlength{\itemsep}{0pt}
			\setlength{\parsep}{1pt}
			\setlength{\topsep}{1pt}
			\setlength{\partopsep}{0pt}
			\setlength{\leftmargin}{1.5em}
			\setlength{\labelwidth}{1em}
			\setlength{\labelsep}{0.5em} } }
	
\newcommand{\squishend}{\end{list}}
\usepackage{multirow}

\usepackage{amsthm}
\usepackage{mathrsfs}

\lstdefinestyle{SQLStyle}{
  language=SQL,
  showspaces=false,
  basicstyle=\ttfamily\footnotesize,
  commentstyle=\color{gray},
  mathescape=true,
  numbers=none,
  escapeinside={^}{^},
  captionpos=b,
  mathescape=false,
}

\newcommand{\mypython}[1]{%
  \begin{tcolorbox}[
    colback=gray!10,
    colframe=black,
    boxrule=0.5pt,
    arc=2pt,
    left=5pt,
    right=5pt,
    top=2pt,
    bottom=2pt
  ]
  \small\ttfamily #1
  \end{tcolorbox}
}

%% file: abstract.tex

\begin{abstract}
Many {\em templatized documents} are programmatically generated from structured data following a visual template. Such documents include invoices, tax documents, financial reports, and purchase orders.  Effective data extraction from these documents is crucial to support downstream analytical tasks. 
Current data extraction tools often struggle with complex document layouts, incur high latency and/or cost on large datasets, and require significant human effort. 
The key insight of our tool, \sys, 
is to infer the underlying template 
used to create such documents, \update{and then extract the data, rather than extracting directly from documents. }
To do so, \sys first infers the underlying fields, 
such as columns of tabular portions 
or keys in co-located key-value pairs, 
by leveraging their consistent location 
patterns (e.g., two fields
in the same template repeatedly co-occur
within a fixed distance apart across multiple records). 
\sys then assembles these fields into a template 
by enforcing visual constraints,
such as vertically 
aligning table rows with their column headers
for tabular regions, and 
horizontally aligning keys with their values
for key-value pairs. 
\sys then uses this inferred template to 
accurately and efficiently extract data from templatized documents 
at a low cost. 
\yiming{On one benchmark with 34 diverse real-world datasets,} \sys outperforms state-of-the-art structured data extraction tools (Evaporate, Textract, and Azure Document Intelligence), and vision-based LLMs like GPT-4-Vision, by over 25\% in precision and recall. \yiming{Another benchmark with 30 large datasets demonstrates \sys 's scalability:} 
\update{it is 520$\times$ faster and 3,786$\times$ cheaper than the most competitive compared tool}, for extracting data from large document collections with over 2000 pages. 
\end{abstract}


%% file: introduction-Aditya-March28.tex

\section{Introduction}
\label{sec:introduction}


Documents (including PDFs or Word files)
are ubiquitous across organizations, big and small,
and represent some of the most valuable---and yet untapped---sources 
of insight in
data lakes.
A particularly common class of documents
are those that are {\em templatized}:
documents that are programmatically  
generated at scale from structured data by 
employing a {\em visual template}, i.e., visually rendering
each field from a given record
in a specific location on the document, and repeating
this process for each record---see examples in Figures~\ref{fig:police-example} and \ref{fig:invoice}.
This includes tax forms, invoices, financial reports,  
pay stubs, certification records, expense reports, and purchase orders.
Since these documents are programmatically generated, 
it is easy to
generate many pages at scale 
by simply rendering individual records from one or more relations visually. 
Indeed, considering just invoices
alone, recent estimates forecast over 
0.5 trillion invoices generated annually~\cite{invoice}. 
So, {\em given a collection
of templatized documents, can we cheaply, efficiently, 
and accurately extract all structured data from it?}
Doing this will help us effectively 
perform Extract-Transform-Load (ETL) from such document collections,
with the corresponding structured data outputs stored in data warehouses for subsequent analysis.

\begin{figure*}[tb]
    \centering
    \includegraphics[width=1\linewidth]{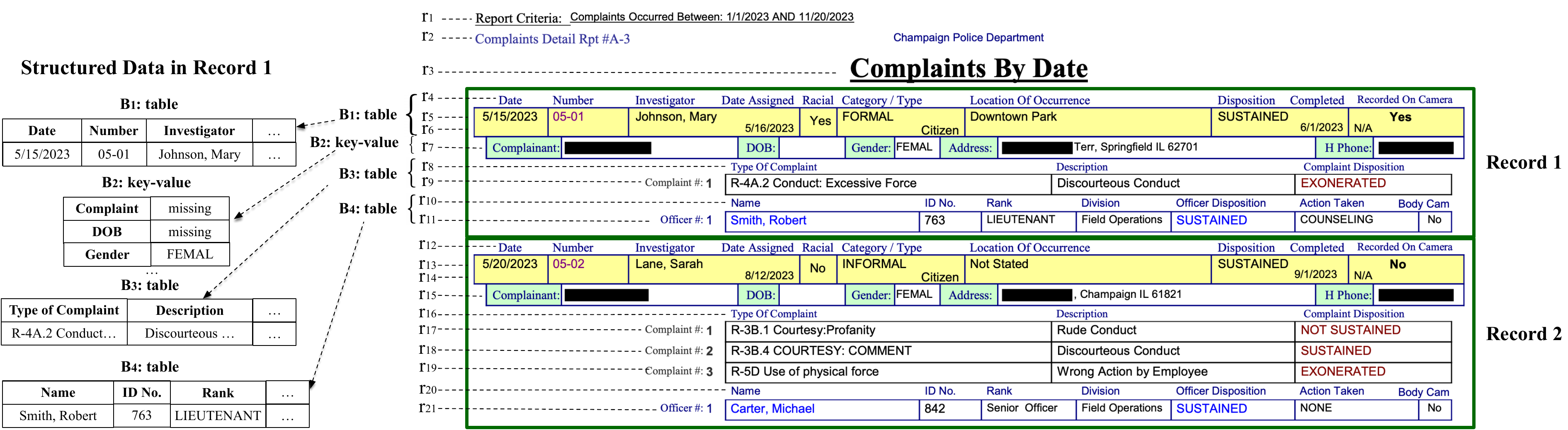}
    \vspace{-1.5em}
    \caption{\small Police complaint records provided by our collaborators. (Actual values have been replaced for privacy.) Row indices $[r_1, r_2, \dots]$ and block labels $[B_1, B_2, \dots]$ are manually annotated, where $B$ stands for a data block. }
    \vspace{-1em}
    \label{fig:police-example}
\end{figure*}

\topic{\rthree{Document data extraction is still a pain point}}
Intuitively, the sheer heterogeneity of templatized documents 
presents challenges for automated approaches.
Templatized documents typically 
contain a mix of tabular and key-value blocks.
Figure~\ref{fig:police-example} shows an example from police complaint records from our collaborators\anonymous{\footnote{Actual name replaced for anonymity.}}{ Big Local News at Stanford University};
here, record 1 and 2 are visually similar 
because they were generated using the same template. This example, as in the invoice example in Figure~\ref{fig:invoice}, contains table and key-value blocks,
with each requiring a distinct approach for data extraction.
In a table block, fields are in the first row (e.g., $\{$\code{Date}, \code{Number}, ...$\}$ in Figure~\ref{fig:police-example}), with corresponding values in each column.  A key-value block usually places fields and their corresponding
values horizontally within a row (e.g., \code{Gender} and \code{Female}). 
Furthermore, different data blocks may be arranged in a nested structure. In Figure~\ref{fig:invoice}, multiple small tables (e.g.,  $B_3$ and $B_4$) are nested within a larger table, $B_2$.  
Thus, to interpret each document, {\bf \em we need to reason about vertical
alignment} {\em (of column names with corresponding values per record in a table)}, as well as
{\bf \em horizontal alignment} {\em (of values per column associated with the same record for a table,
or of keys and their corresponding co-located values in a key-value block)}.

\begin{figure}[tb]
    \centering
    \includegraphics[width=0.8\linewidth]{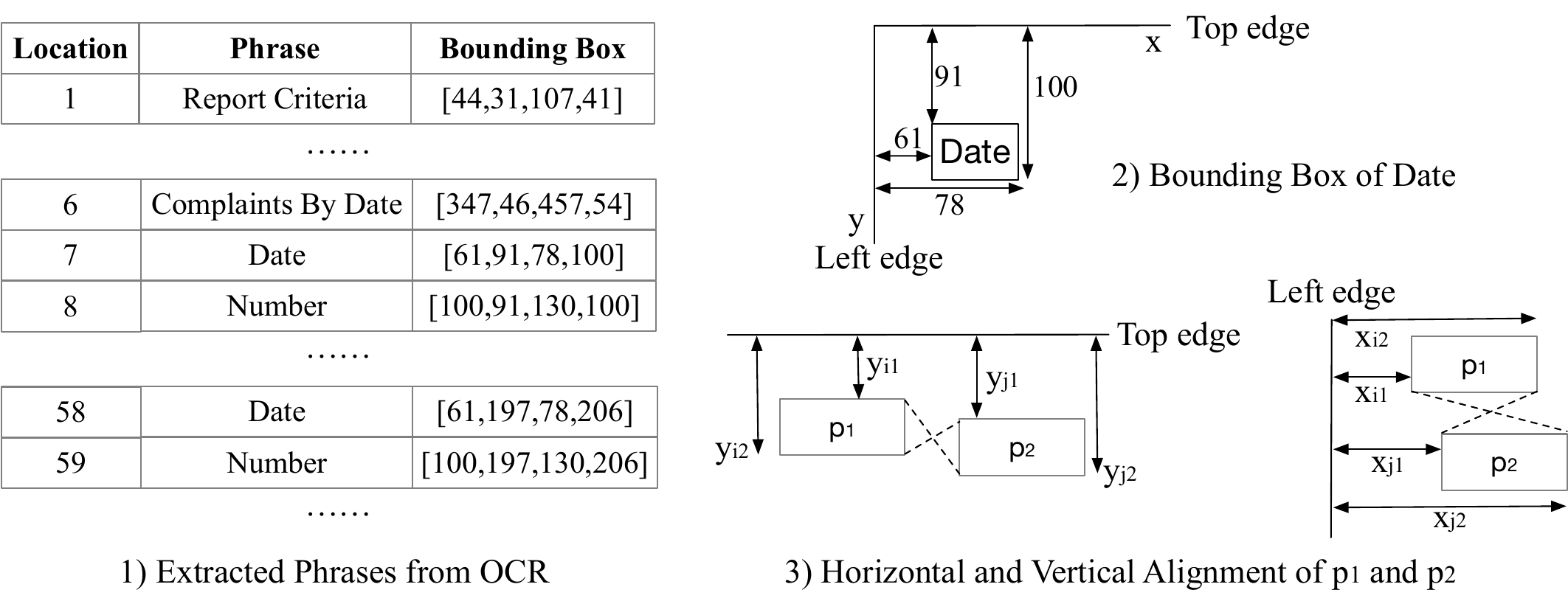}
    \vspace{-1em}
    \caption{\small \rthree{Extracted Phrases with Locations and  Bounding Boxes in Police Complaint Records. }}
    \vspace{-1em}
    \label{fig:phrases}
\end{figure}

Given the heterogeneity, complexity, and need for reasoning about both horizontal and vertical visual alignments, how can we extract structured data from templatized documents?
Past work from the 2000s---mainly from the database community---considered 
HTML-based extraction (e.g., \cite{parameswaran2011optimal,arasu2003extracting,dalvi2009robust}),
which relies on HTML tag hierarchies
that are absent from  documents, 
and general information extraction (e.g., \cite{cafarella2009web, sarawagi2008information, etzioni2004web, agichtein2000snowball}),
which operates on free text without considering visual structure.
Recent efforts on document extraction~\cite{or2021few,tata2021glean,parthasarathy2022landmarks,xu2020layoutlm,aggarwal2021form2seq,sarkhel2021improving} often require users to specify fields and provide labeled examples, leading to high human effort,
and additionally lose relationships between extracted values across fields. 
\techreport{For example, values for \textit{Date} and \textit{Number} in Figure~\ref{fig:police-example} would be extracted independently, without capturing whether the extracted values belong to the same table or record. Furthermore, fields in the document (e.g., \code{Start Date} in Figure~\ref{fig:invoice}) may not be unique, leading to ambiguous results.}
Other approaches, employed by recent document processing systems~\cite{shankar2024docetl,lin2024towards,patel2024lotus,liu2024declarative}, involve vision-based LLMs (e.g., GPT-4 Vision~\cite{gpt4vision})
that can leverage visual features,
a text-based LLM operating on OCR (Optical Character Recognition) output
of a document, 
or pretrained document extraction APIs, such as AWS Textract~\cite{textract} or Azure Document Intelligence~\cite{azure}.
However, these approaches, including those that build on them~\cite{arora2023language},
struggle with complex layouts, achieving only 25\%--65\% precision and recall 
on a benchmark of 34 real-world datasets. 
Most approaches process documents 
page by page, resulting in high latency and cost, e.g.,  GPT-4 Vision takes 30+ hours and 
\$50+ to extract data from 2000+ pages. 
We discuss these approaches in detail in Section~\ref{sec:relatedwork}.

\topic{\sys: Inferring templates prior to extraction} 
We present \sys\footnote{\sigmod{Short for {\bf T}emplatized document {\bf W}rangling for {\bf I}nformation e{\bf X}traction.} \arxiv{Short for {\bf T}emplatized info{\bf R}mat{\bf I}on e{\bf X}traction.}}, a tool for extracting data
from templatized documents that first  {\em reconstructs the template from a set of documents generated using the same template}, and then uses the reconstructed template to extract data from the documents. 
Unlike approaches that apply LLMs or pretrained APIs to each page in a document collection, once the template is inferred, \sys no longer requires such calls and can extract structured data quickly, accurately, and at no cost. This is because the complex document can be decomposed, based on the template, into data blocks with simple structures (tables or key-value pairs). In Figure~\ref{fig:police-example}, every record, like Record~1,  can be decomposed into four blocks: a table in block $B_1$, a list of key-value pairs in $B_2$, followed by two tables in $B_3$ and $B_4$. Extracting data from these simple blocks, once identified, is easier than operating directly on documents. 
Moreover, extraction using templates preserves data relationships \yiming{(e.g., which extracted values belong to the same row)}  without requiring users to specify attributes or provide labeled data.

\rone{A natural question that emerges is: are visual templates already available? Unfortunately, they are often absent in real-world settings. Document creators are typically not the ones analyzing them. For example, our journalism collaborators analyze use-of-force records in Figure 1 produced by police departments, while buyers analyze invoices generated by sellers. 
Specifying templates is also challenging for non-technical users due to the documents’ mixed and heterogeneous  structure~\cite{chasins2018rousillon}. In Figures~\ref{fig:police-example} and~\ref{fig:invoice}, users may struggle to distinguish between tables and key-value pairs or incorrectly flatten nested tables. Even simply specifying fields is difficult: our datasets contain up to 58 fields per record, making manual specification time-consuming and error-prone. }

\topic{\rthree{Template inference remains challenging}}
While inferring a visual template
and then using it for extraction intuitively makes sense, template inference presents challenges in two fronts:  {\em field inference}, 
i.e., inferring which phrases are fields,
and {\em template assembly}, i.e., assembling the 
fields into a template used
to generate records. 
Fields refer to columns in tables, e.g., \code{Date} in $B_1$, or keys in key-value pairs, e.g., \code{DOB} in $B_2$.  
Inferring fields in templatized documents is non-trivial due to complex document layouts. 
Real-world documents often contain a mix of nested key-value and tabular blocks, interspersed with metadata (e.g., headers, titles). 
While one could use visual indicators 
that are easy for a human to understand, e.g., vertical or horizontal alignment, proximity of value phrases to key ones, 
or indentation, these are hard to automatically leverage and develop rules for. 
Moreover, inferring fields alone is insufficient to capture the template, which must also encode the structural pattern, i.e., how fields are organized to form records. 

\begin{figure}[tb]
    \centering
    \includegraphics[width=0.6\linewidth]{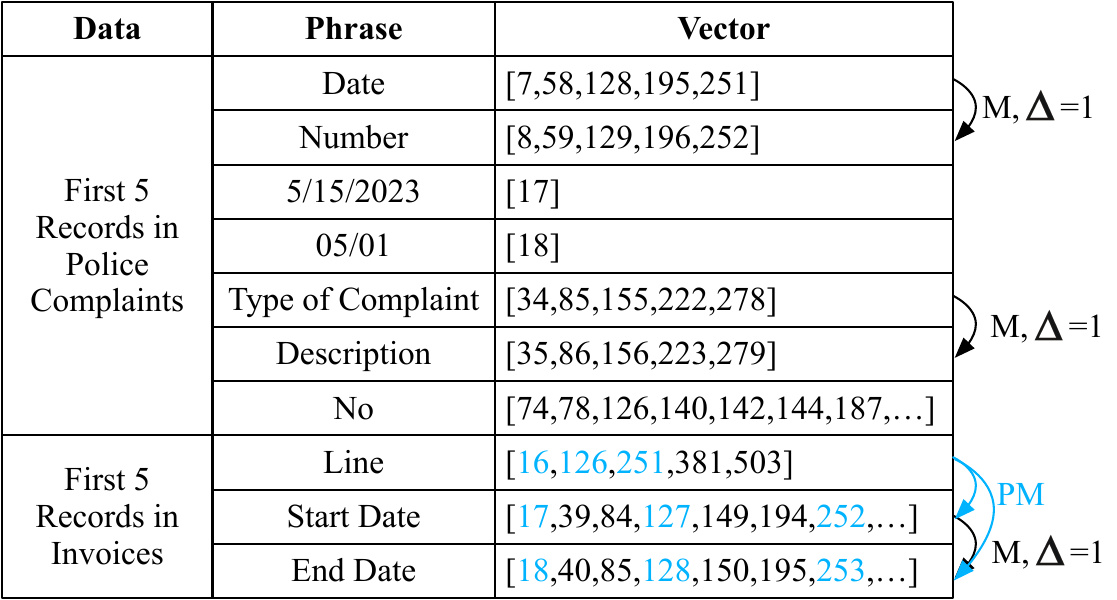}
    \vspace{-1em}
    \caption{\small Location Vectors of Sample Phrases in First 5 Records in Police Complaints and Invoices. M and PM denote Perfect Match and Partial Perfect Match. }
    \label{fig:vector}
\end{figure}
\begin{figure*}[tb]
    \centering
    \includegraphics[width=1\linewidth]{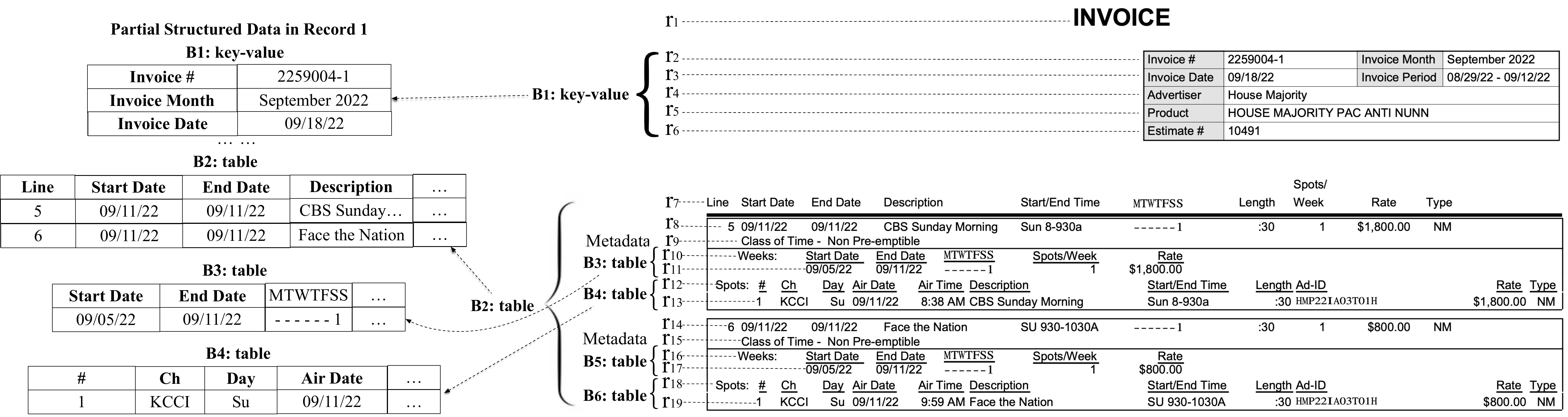}
    \vspace{-2em}
    \caption{\small Portions of Record 1 in an Invoice Document from the Open Benchmark. Row indices $[r_1, r_2, \dots]$ and block labels are manually annotated. There are thousands of invoice records in the document that follow the same template.} 
    \label{fig:invoice}
    \vspace{-1.2em}
\end{figure*} 

\topic{\rthree{Inferring fields based on visual patterns and LLMs}} Our key insight is that {\em fields tend to appear in similar locations across records} (e.g., $\{$\code{Date}, \code{Number}, $\dots$ $\}$ in Figure~\ref{fig:police-example}), 
while {\em the values for the same field 
can be different across records} (e.g., \code{05-01} and \code{05-02} are values of \code{Number} in Record 1 and 2, respectively). 
Consider the phrases extracted by a typical OCR tool~\cite{pdfplumber} in Figure~\ref{fig:phrases}-1, where the location indicates order of extraction (e.g., \code{Report Criteria} and \code{Date} in Record 1 are the 1st and 7th phrases). If we then assemble every location where a given phrase appears into a vector as in Figure~\ref{fig:vector}, for fields like \code{Date} and \code{Number}, the difference (or gap $\Delta$) in their location vectors  is constant, since \code{Number} is always the next extracted phrase to \code{Date} in every record. Among all phrases, typically only a small subset (corresponding to fields) exhibit the consistent location pattern as described above with constant differences (i.e., $\Delta$). Instead, most values (e.g., \code{05-01}) appear randomly across records and their location vectors are often not consistent.  
We propose a clustering approach to group phrases with consistent location patterns. When visual cues prove insufficient (e.g., the value for a field is constant across records), we use the semantic knowledge of LLMs to improve robustness. 
Instead of relying on LLMs to handle complex tasks, such as recognizing intricate visual layouts,
we restrict their role to  simple questions, like "Is this phrase a field or a value?".  

\topic{\rthree{Assembling template based on row labels}} To further assemble the template based on inferred fields, 
our key insight is that a record typically 
{\em consists of multiple data blocks, such as table or key-value blocks, and their placement follows a consistent pattern} across records. To infer the template, \sys assigns each row a label from $\{${\em Key}, {\em Key-Value}, {\em Value}, {\em Metadata}$\}$, each associated with a probability estimated based on the inferred fields. We cast the row labeling problem as one that finds the most probable label per row, constrained by the validity of  table structure. For example,  a {\em Key} (e.g., row $r_4$ in Figure~\ref{fig:police-example}) row must have a vertically-aligned {\em Value} row (e.g., $r_5$) beneath it, and vice versa. While {\em Key-Value} rows must have a key for each value, some keys could have missing values. 
We formalize this {\em row labeling problem} as an Integer Linear Programming problem  (ILP), and prove its {\bf NP-hardness}.  We then design efficient pruning strategies on a small number of rows to infer the row labels and then the template. 

Once the template is inferred, completing data extraction is straightforward. \sys divides templatized document into a list of records, and decomposes each record into a list of data blocks with simple structures (i.e., tables or key-value pairs), each of which is easy to extract data from. Overall, we make the following contributions as part of  \sys, a robust tool for data extraction from templatized documents. 


\begin{itemize}[leftmargin=*]
    \item We introduce the concept of a {\em template} and propose novel algorithms to infer the template for data extraction. Specifically, we develop a clustering-based algorithm that reasons about location patterns, along with LLM input, to infer fields. We further formalize the problem of inferring the template structure as an ILP, show its NP-hardness, and provide an efficient solution.  
    \item We present efficient techniques to extract structured data based on the inferred template at no additional cost. 
    \item We conduct a comprehensive evaluation on two  benchmarks comprising 64 diverse real-world datasets, comparing our approach against six state-of-the-art techniques. \sys achieves around 90\% precision and recall, {\bf outperforming each baseline by over 25\%} in both metrics. We introduced a new metric called structure precision and recall that accounts for the nested and heterogeneous structure, and once again,   \sys \hspace{0.3mm} {\bf outperforms the best baseline by over 31\% (and 22\%) on structure precision (and recall)}.  \sys scales easily to large datasets and is \textbf{520$\times$} \textbf{faster} and \textbf{3786$\times$} \textbf{cheaper} than the most competitive baseline, on document collections with 2000+ pages. 
\end{itemize}

%% file: overview.tex
\section{Background and Overview}
\label{sec:background}


We focus on a single document $D$, formed by concatenating all documents $D_1, D_2, \dots, D_n$ that share the same template. Then $D$ consists of records created using the same template (e.g., Records 1 and 2 in Figure~\ref{fig:police-example}) and other metadata (e.g., headers, footers),  where the concepts of record, template, and metadata are defined shortly. This setting is common, e.g., invoices, purchase orders, tax documents, financial reports, and immigration forms. 


\topic{Phrases} 
We operate on the serialized plain text representation of the concatenated document $D$. Let $P$ be the phrases extracted from $D$ by using Optical Character Recognition (OCR) tools~\cite{pdfplumber,pymupdf},\footnote{Our approach is suitable for documents where phrases and their bounding boxes can be extracted, such as PDFs, Word documents, and  images of scanned documents.} in ascending order of location, i.e., $P = [p_1,p_2,...,p_m]$. Figure~\ref{fig:phrases}-1 presents phrases extracted from the police complaint document in Figure~\ref{fig:police-example}. 
Each phrase $p_i$ is paired with its bounding box $b_i = [x_{i1},y_{i1},x_{i2},y_{i2}]$ shown in Figure~\ref{fig:phrases}-1. Here, $x_{i1}$ and $x_{i2}$ represent distances from the left and right edges of $p_i$ to the left edge of the page, respectively. Similarly, $y_{i1}$ and $y_{i2}$ denote distances from the top and bottom edges of $p_i$ to the top edge of the page, respectively. Figure~\ref{fig:phrases}-2 presents the bounding box of the first \code{Date} (\code{Date} in Record 1). Each phrase $p_i$ has its {\em location} $i$ in the document, denoted as $loc(p_i)$, determined by the OCR tool, which extracts phrases row-by-row from left to right, as shown in Figure~\ref{fig:phrases}-1. For instance, \code{Report Criteria} and \code{Date} (in Record 1 in Figure~\ref{fig:police-example}) are the first and seventh phrases extracted with location $1$ and $7$, respectively. \rone{\sys supports both well-formed and scanned documents. In our Q-Benchmark, scanned (well-formatted resp.) documents occupy 22\% (78\% resp.). \sys can handle any document as long as the OCR tool can extract phrases and their bounding boxes.} 

\topic{Rows and Blocks} 
A row is a list of horizontally aligned phrases.  \rthree{Two phrases $p_i$ and $p_j$ (e.g., $p_1$ and $p_2$ in Figure~\ref{fig:phrases}-3),  with bounding boxes $[x_{i1}, y_{i1}, x_{i2}, y_{i2}]$ and $[x_{j1}, y_{j1}, x_{j2}, y_{j2}]$, are visually {\em horizontally aligned}, if $y_{i1} \leq y_{j2} \land y_{i2} \geq y_{j1}$, and are {\em vertically aligned}, if $x_{j1} \leq x_{i2} \land x_{j2} \geq x_{i1}$.} We then  transform phrases $P = [p_1, p_2, \dots]$ to rows $R = [r_1, r_2, \dots]$ greedily by scanning $p_i \in P$ in increasing order of location. A phrase $p_i$ is merged into an existing row $r$ if $p_i$ is horizontally aligned with every phrase in $r$; otherwise, a new row is created for $p_i$. We show all rows in Figure~\ref{fig:police-example}. \techreport{When a phrase $p$ is horizontally aligned with multiple rows, such as ``\code{Yes}'', the first value under column \code{Racial} in Figure~\ref{fig:police-example}, $p$ is merged into the first row it is horizontally aligned with, e.g., \code{Yes} is part of $r_5$ instead of $r_6$.} Finally, a block $B = [r_i, r_{i+1}, \dots, r_j]$ is a list of rows in $D$.

\begin{figure}[tb]
    \centering
    \includegraphics[width=1\linewidth]{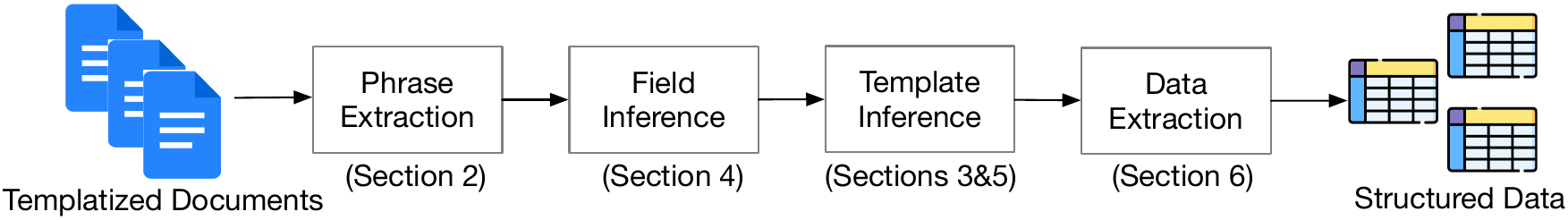}
    \vspace{-2em}
    \caption{\small Overview of the \sys Pipeline. }
    \vspace{-1.5em}
    \label{fig:pipeline}
\end{figure}

\topic{Phrase Labels}
We assign each phrase $p \in P$ one of the $\{field, value, metadata\}$ labels. {\em Fields} refer to the column names of tables or keys in key-value pairs. A {\em value} is a phrase whose corresponding field can be identified, such as a cell in a table row, or a value in a key-value pair. Note that the same phrase may appear multiple times in $P$ (e.g., \code{Date} in police complaints appears in row $r_4$ and $r_{12}$). 
Other non-field phrases whose fields cannot be identified are called \textit{metadata}, such as the title ``\code{Complaints By Date}'', or headers, e.g., phrases in row $r_1$ and $r_2$ in Figure~\ref{fig:police-example}. \techreport{We assume that the label for each phrase $p \in P$ is unique. For instance, a phrase cannot be both a field and the value of another field, a scenario that we rarely observe in real-world datasets.}

\topic{\rthree{Row and Block Labels}} 
\rthree{A row $r_i$ is assigned one of the $\{${\em Key}, {\em Key-Value}, {\em Value}, {\em Metadata}$\}$ labels. $r_i$ is a {\em Key} row if it contains no values and each field $p \in r_i$ has at least a vertically aligned value $p'$ beneath $p$. Similarly, a {\em Value} row $r_i$ contains no fields, and each value has a vertically aligned field preceding it. $r_i$ is a {\em Metadata} row if any phrase in $r_i$ is metadata, and in a {\em Key-Value} row, any value has a preceding field. } 
\rthree{We now define a {\em table} or a {\em key-value} block. A block $B = [r_i, r_{i+1}, \dots, r_j]$ is a table block if $r_i$ is a Key row and other rows in $B$ are all Value rows, where for any value $p \in r_j$ with $j > i$, there exists a field $p' \in r_i$ such that $p$ is vertically aligned with $p'$. $B$ is a key-value block if any row $r_i\in B$ is a Key-Value row.}




 
\topic{\update{Overall \sys pipeline}} \update{Figure~\ref{fig:pipeline} shows an overview of \sys. {\bf \circled{1} Phrase Extraction}: Given a concatenated templatized document $D$, \sys applies OCR tools to convert $D$ into a list of phrases, each paired with its bounding box and location (e.g., part of the extracted phrases in police complaint records is shown in Figure~\ref{fig:phrases}-1). {\bf \circled{2} Field Inference}: \sys clusters the extracted phrases to group phrases that share consistent location patterns into fields, and uses LLMs to provide semantic knowledge and enhance robustness (Section~\ref{sec:field-prediction}).  {\bf \circled{3} Template Inference}: We model a template as a tree (Section~\ref{sec:template}). Given the inferred fields, \sys infers the row labels as one of $\{${\em Key}, {\em Key-Value}, {\em Value}, {\em Metadata}$\}$.  
We solve the row labeling problem using Integer Linear Programming, with the goal to find the most probable label assignments constrained by visual alignments 
(Section~\ref{sec:template-prediction}). {\bf \circled{4} Data Extraction}: We use the inferred template to extract data from the entire document $D$  (Section~\ref{sec:data-extraction}).} \sys is implemented in 1200+ lines of Python. 


%% file: template_formulation.tex
\section{Template Formalization}
\label{sec:template}

We now formally define templates.
A document $D$ consists of a set of {\em records}, each generated using a {\em template}. Each record recursively consists of {\em blocks}. 
Let $T$ be the {\em template} used to generate records, where $T=(V,E)$ is an ordered directed tree with an artificial root node. Each non-root node  $v\in V -  \{root\}$ is associated with a \textit{type} $\in \{Table$, $Key$-$Value\}$, and a set of \textit{fields}.
In Figure~\ref{fig:template}, we present the templates for police complaints and invoices, respectively, where the type and fields for each node are specified. Intuitively, a template defines how a document is populated with records by specifying which fields are populated, in what manner (e.g., as a table or a key-value block), and order (e.g., which block appears first).

For a node $v_i \in V -  \{root\}$, we denote $v_i \rightarrow B_i$ to be the process of generating a data block $B_i$ by populating fields in $v_i$. $B_i$ is a sequence of rows, where every {\em field}  in  $v_i.fields$ appears exactly once in $B_i$ and all of the {\em value} phrases that correspond to those field phrases appear in $B_i$. 
Let $Rec_i$ be a record created from template $T$, comprising a sequence of data blocks generated by nodes in $T$. To do so, each node $v_i$ in $T$ generates one or more data blocks in the predefined order in $T$. 

\begin{example}
Block $B_1$ in police complaint record $Rec_1$ corresponds to $[r_4,r_5,r_6]$ in Figure~\ref{fig:police-example}. In its template  in Figure~\ref{fig:template}-1, $B_1$ is generated by populating the fields from the table node $v_1$ in Figure~\ref{fig:template}-1. A single record then 
comprises one or more data blocks generated by each of $v_1$, $v_2$, $v_3$ and $v_4$ in order. In Figure~\ref{fig:invoice}, $B_3$ and $B_5$ in invoice documents are both generated by $v_3$ in Figure~\ref{fig:template}-2, while $v_2$ generates $B_2$. Note that $v_2$ is the parent of $v_3$, indicating that $v_3$'s data blocks $\{B_3,B_5\}$ are nested within $v_2$'s block $B_2$, as we will discuss below. 
\end{example}

\begin{figure}[tb]
    \centering
    \includegraphics[width=1\linewidth]{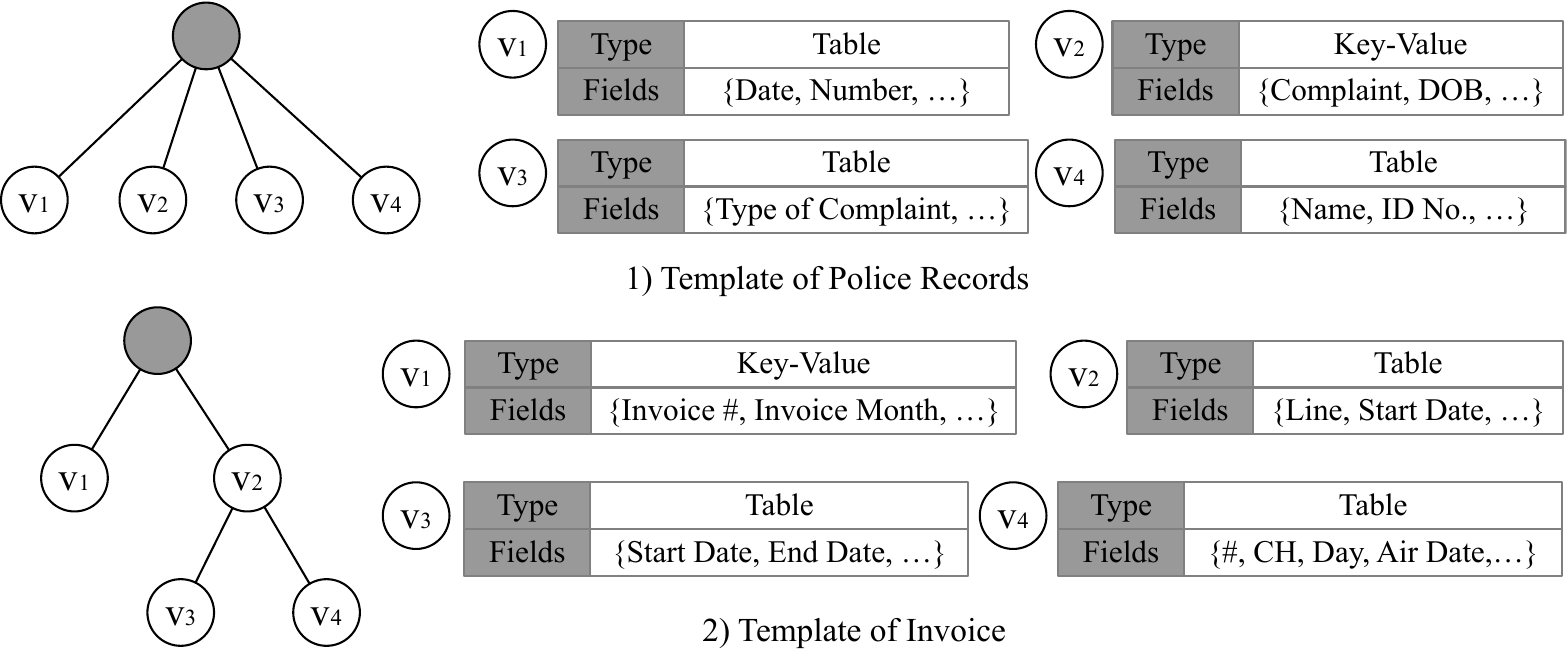}
    \vspace{-2em}
    \caption{\small Templates for Police Records and Invoice Document. }
    \vspace{-1.5em}
    \label{fig:template}
\end{figure}

Let $\mathcal{B}_v^{Rec}$ be the set of blocks generated by the node $v$ in the record $Rec$, i.e., $\mathcal{B}_v^{Rec} = \{B_i | v \rightarrow B_i, B_i \in Rec\}$. In the police 
complaint record in Figure~\ref{fig:police-example}, $\mathcal{B}_{v_1}^{Rec_1} = \{B_1\}$, while $\mathcal{B}_{v_3}^{Rec_1} = \{B_3,B_5\}$ in the invoice document in Figure~\ref{fig:invoice}. Let $loc(\mathcal{B}_v^{Rec})$ be the location of $\mathcal{B}_v^{Rec}$, defined as the smallest (earliest) location of  phrase $p$ in any block in $\mathcal{B}_v^{Rec}$, i.e., $loc(\mathcal{B}_v^{Rec}) = \min{loc(p)}, p\in B_i, \forall B_i \in \mathcal{B}_v^{Rec}$.

Two data blocks may overlap (e.g., $B_2$ and $B_3$ in Figure~\ref{fig:invoice}) \update{if their visual bounding boxes overlap}, while others may not (e.g., $B_1$ and $B_2$ in Figure~\ref{fig:police-example}). To formally define overlapping relationships between data blocks, for two blocks $B_i$ and $B_j$, we denote $B_i \cap B_j = \emptyset$ if $\forall p_{i_1}, p_{i_2} \in B_i$ and $p_{j_1}, p_{j_2} \in B_j$, such that $(i_1 > j_1) \oplus (i_2 > j_2) = 0$, where $i_1,i_2,j_1,j_2$ are phrase  locations, and $\oplus$ denotes a logical XOR. Otherwise, $B_i \cap B_j \neq \emptyset$. Intuitively, $B_i \cap B_j = \emptyset$ implies that the visual bounding boxes of the two blocks do not overlap, such as any two blocks in  police complaint records. In the invoice records in Figure~\ref{fig:invoice}, $B_1$ $\cap$ $B_2 = \emptyset$, while $B_2 \cap B_3 \neq \emptyset$ since phrases from $B_2$ are placed both before and after $B_3$, e.g.,  row $r_8$ appears before $B_3$, while row $r_{14}$ appears after $B_3$.
Given a record $Rec$, and $v_i,v_j \in V$, we denote $\mathcal{B}_{v_i}^{Rec} \cap \mathcal{B}_{v_i}^{Rec} = \emptyset$ if $\forall B_i \in \mathcal{B}_{v_i}^{Rec}, B_j \in \mathcal{B}_{v_j}^{Rec}$, $B_i \cap B_j = \emptyset$.

Now we describe how edges in the template $T$ influences overlap of data blocks.  For any two nodes $v_i, v_j \in V - \{root\}$, when $v_i$ is an {\em ancestor} of $v_j$, then $\forall Rec ,loc(\mathcal{B}_{v_i}^{Rec}) < loc(\mathcal{B}_{v_j}^{Rec})$ and $\mathcal{B}_{v_i}^{Rec} \cap \mathcal{B}_{v_i}^{Rec} \neq \emptyset$. Likewise, when $v_i$ is a {\em left sibling} of $v_j$,  then $\forall Rec, loc(\mathcal{B}_{v_i}^{Rec}) < loc(\mathcal{B}_{v_j}^{Rec})$ and $\mathcal{B}_{v_i}^{Rec} \cap \mathcal{B}_{v_i}^{Rec} = \emptyset$.




\begin{example}
In Figure~\ref{fig:template}-1,  the non-root nodes are placed at the same level from left to right, indicating that their corresponding data blocks $B_1$, $B_2$, $B_3$, and $B_4$ for the given record do not overlap.  In contrast, the template in Figure~\ref{fig:template}-2 shows that node $v_2$ is an ancestor of $v_3$, implying that the data blocks of $v_3$, $\mathcal{B}^{Rec}_{v_3} = \{B_3, B_5\}$, overlap with the data blocks of $v_2$, $\{B_2\}$, and $B_2$ appears before $B_3$ and $B_5$. Note that the sequence of data blocks corresponding to the subtree rooted at $v_2$ can be repeated, provided the sequence of data blocks for its children complies with the edge definition above.
\end{example}


Given the concepts of templates, data blocks, and records, we next describe how \sys infers the template (Section~\ref{sec:template-prediction}) by first identifying fields (Section~\ref{sec:field-prediction}), and then extracts data (Section~\ref{sec:data-extraction}).

%% file: field_prediction.tex
\section{Field \update{Inference}}
\label{sec:field-prediction}


\sys first infers a set of fields given the extracted phrases $P$ from the document $D$, as the set of fields often appear in similar locations across records in templatized documents.

\vspace{-2mm}
\subsection{Location Vectors and Matches}
Since the same phrase $p$ (e.g., \code{Date} in Figure~\ref{fig:phrases}-1) may appear multiple times in $P$, such as $p_7, p_{58}$, we denote $v_p = [i, \dots, j]$ as the {\em location vector} of $p$, comprising the locations of occurrences of $p$ in ascending order. Figure~\ref{fig:vector} lists the location vectors for sample phrases in the first five records in police complaints and invoices. \code{Date} appears exactly once in each record, corresponding to the list of phrases $[p_7, \dots, p_{251}]$, and thus has the location vector $[7, \dots, 251]$. Let $L_p$ be the length of the vector $v_p$. 

To formalize the intuition that two related fields $p_i$ and $p_j$   share similar locations, their location vectors are often {\em related by a constant shift}. In Figure~\ref{fig:vector}, adding one to each entry in the location vector for \code{Date} aligns it with the location vector for \code{Number}. This is because \code{Date} and \code{Number} are fields in the same table node, and thus appear in sync in the blocks corresponding to that node, resulting in constant relative distances across records. We define the concept of {\em perfect match} to capture this observation.


\begin{definition}[Perfect Match.]
\label{def:perfect-match}
     Let $v_{p_i} = [i_1,i_2,...]$ and $v_{p_j} = [j_1,j_2,...]$ be the location vectors of phrases $p_i$ and $p_j$. We say $p_i$ is a {\em perfect match} with $p_j$, denoted by $M(v_{p_i},v_{p_j}) = 1$, if $L_{p_i} = L_{p_j}$, $L_{p_j} > 1$ and $\exists \Delta$, s.t., $\forall i_k \in v_{p_i}, j_k \in v_{p_j}, |i_k - j_k| = \Delta$. 
\end{definition}

When $L_{p_j} = 1$, any two phrases will be a perfect match, and thus we enforce $L_{p_j} > 1$ above. \update{When a field (e.g., \code{Start Date} in Figure~\ref{fig:invoice}) appears in multiple data blocks ($B_2,B_3$) within a record, its location vector may not be perfectly aligned with that of another field (e.g., \code{Line} in $B_2$) that appears in a single data block, even though both fields also share similar location patterns.} Thus, we further relax this notion below. 



\begin{definition}
    [Partial Perfect Match.] Let $v_{p_i} = [i_1,i_2,...]$ and $v_{p_j} = [j_1,j_2,...]$ be the location vectors of phrases $p_i$ and $p_j$. We say $p_i$ is a partial perfect match with $p_j$, denoted by $PM(v_{p_i},v_{p_j})=1$, if there exists a subsequence $v_{p_i}'\subseteq v_{p_i}$, $M(v_{p_i}',v_{p_j})=1$. 
\end{definition}


\begin{example}
\label{eg:match}
    In Figure~\ref{fig:vector},  any pair of phrases among  \code{Date}, \code{Number}, \code{Type of Complaint}, and \code{Description} is a perfect match, while \code{5/15/2023}, \code{05/01} and \code{No} are not a perfect match with any other phrases. In Figure~\ref{fig:invoice}, the same field (e.g.,  \code{Start Date}) may appear in multiple blocks (e.g., $B_2$, $B_3$, and $B_5$) within a record, while some phrases, such as \code{Line}, only appear in one block (e.g., $B_2$). In this case, \code{Line} is a partial perfect match with \code{Start Date} because a subsequence of the location vector for \code{Start Date} perfectly matches that of \code{Line}. This subsequence corresponds to the appearances of \code{Start Date} in $B_2$ every record. 
\end{example}

\vspace{-1mm}
\subsection{Analysis of Location-based Matches}
\label{subsec:location-analysis}
\update{We have shown previously that similar location patterns between fields can be captured by (partial) perfect matches based on their location vectors.}  Intuitively, most location vectors are ``irregular'', while a small amount of them are ``regular'' (i.e., sharing perfect or partial matches). Informally, regular vectors are fields that tend to co-occur together (i.e., appear in documents near each other), whereas irregular vectors correspond to values or metadata that occur randomly. \update{Now, we establish two properties of the match functions above to show their effectiveness in inferring fields.}

Given document $D$, let $T' = (V',E')$ be the \textit{true} template for $D$, and let $F'$ be the corresponding set of {\em fields}. 

\begin{proposition}
\label{prop:table}
    Given the true template $T' = (V', E')$, when there exists a unique node  $v\in V', v.type = table$, $p_i,p_j \in v.fields$, if $L_{p_i} = L_{p_j}$, then $p_i$ and $p_j$ are a perfect match; if $L_{p_i} > L_{p_j}$, then $p_j$ and $p_j$ are a partial perfect  match. 
\end{proposition}

Proposition~\ref{prop:table} (proof in Appendix~\ref{subsec:proof-proposition1}) states that if phrases $p_i$ and $p_j$ are both fields of a unique table node in $T'$ (e.g., \code{Date} and \code{Number} in police complaints), and their location vectors have the same length, they must be a perfect match. This is because records are generated using the same $T'$, ensuring that the relative distance between $p_i$ and $p_j$ remains constant across all records. 
Conversely, when their location vectors have different lengths, such as $L_{p_i} > L_{p_j}$, field $p_i$ (e.g., \code{Start Date} in invoices) may appear in multiple nodes, while $p_j$ (e.g., \code{Line}) belongs only to a node $v$. Here, $v_{p_j}$ and the subsequence of $p_i$'s location vector corresponding to its occurrences in blocks generated by the node $v$, denoted as $v_{p_i}'$, must be a perfect match, since $v_{p_i}'$ and $v_{p_j}$ share the same vector length. 
\techreport{A proof of Proposition~\ref{prop:table} is shown below.}
\techreport{
\begin{proof}[Proof of Proposition~\ref{prop:table}]
In the true template $T' = (V', E')$, when there exists a unique node  $v\in V', v.type = table$, $p_i,p_j \in v.fields$, if $L_{p_i} = L_{p_j}$, we have $\nexists v'\in V'$, $v \neq v'$ , s.t., $p_i \in v'.fields$ or $p_j\in v'.fields$, under the assumption that we make that each phrase $p$ has a unique label. Otherwise, $L_{p_i} \neq L_{p_j}$. Let the location vectors of $p_i$ and $p_j$ be $v_{p_i} = [i_1,i_2,...,i_m]$, $v_{p_j} = [j_1,j_2,...,j_m]$, respectively.   $\forall B_k$, where $v\rightarrow B_k$, let the index of $p_i$ and $p_j$ in $B_k$ be $i_k$ and $j_k$, respectively. $\forall k_1,k_2 \in [1,m]$, we have $i_{k_1} - j_{k_1} = i_{k_2} - j_{k_2}$, since $v\rightarrow B_{k_1}$ and $v\rightarrow B_{k_2}$, and the schema of table in the template is consistent across the records. This completes the first half of proposition. 

When  $L_{p_i} > L_{p_j}$, let $v_{p_i}^{'}$ be the subsequence of $v_{p_i}$ that occurs in blocks created from $v$, i.e., $\forall p_k \in v_{p_i}^{'}, p_k \in B$, where $v\rightarrow B$. Based on the above proof, $v_{p_i}^{'}$ is a perfect match with $v_{p_j}$, and thus $v_{p_i}$ is a partial perfect match with $v_{p_j}$. 
\end{proof}
}

Next we establish the corresponding property for fields in a key-value node. Consider a true field $p$ in a node $v$ whose type is $Key$-$Value$ in $T'$ (e.g., \code{DOB} in police complaints). 
Let $f(p) = True$ if the corresponding value for $p$ is always missing or present in every record generated by $T'$ (e.g., \code{Gender}). Otherwise, $f(p) = False$.



\begin{proposition}
\label{prop:kv}
    Consider the true template $T' = (V',E')$. Given a unique node  $v\in V', v.type = Key$-$Value$, $p_i,p_j \in v.fields$, and $f(p_i)=f(p_j)=True$: if $L_{p_i} = L_{p_j}$, then $p_i$ and $p_j$ are a perfect match; if $L_{p_i} > L_{p_j}$, then $p_j$ and $p_j$ are a partial perfect match. 
\end{proposition}

Proposition~\ref{prop:kv} (proof in Appendix~\ref{subsec:proof-proposition2}) states that if two phrases have a consistent filling pattern ($f(\cdot) = \textit{True}$) in a unique key-value node, they must be a perfect match. For instance, if \code{DOB}'s value is always missing while \code{Gender}'s value is always present in every record, then \code{DOB} and \code{Gender} are a perfect match, provided their location vectors have the same length. Conversely, if $f(p_i) = \textit{False}$, the relative distance between $p_i$ and $p_j$ is not consistent across records. However, even in this case, the phrase $p_i$ with $f(p_i) = \textit{False}$ ends up being a partial perfect match with phrase $p_j$ (where $f(p_j) = \textit{True}$) if there exists a sequence of two data blocks in which $p_i$ is either consistently "present" or consistently "missing," a scenario commonly observed in practice. 

\rmix{Notably, \sys does not rely on bounding boxes (i.e., the physical coordinates of phrases in Figure~\ref{fig:phrases}-2) to infer fields, which may be brittle due to OCR noise. Instead, \sys uses phrase locations based on OCR extraction order to more reliably capture consistent field patterns. 
In Figure~\ref{fig:police-example}, the location difference between \code{Date} and \code{Number} is consistently 1 every record, as \code{Number} is always next to \code{Date}, regardless of their bounding boxes. 
}

\vspace{-1mm}
 \subsection{Field \update{Inference} Algorithm}
 \label{subsec:field-inference-alg}
\update{Location-based matches are effective for \update{inferring} fields as shown by Proposition 1 and 2. However, when a field's value is constant across records (e.g., when the value of \code{Disposition} is always \code{SUSTAINED} in police complaints), this value may share similar location vector patterns with its corresponding field. To address this case, we incorporate LLMs to provide semantic knowledge to help distinguish whether a phrase is  a value or a field.}

We now outline our field \update{inference} algorithm in Algorithm~\ref{alg:field}, \update{which uses location-based matching as the primary method, while incorporating LLMs to enhance robustness and handle edge cases where location vector patterns alone may be misleading. }

\noindent\textbf{Step 1: Phrase Clustering}. \sys first merges any pair of phrases $p_i$ and $p_j$ into one cluster if $M(p_i,p_j) = 1$ (Line 2-9). 
\rtwo{After clustering, true fields that share consistent location vectors are grouped in the same cluster leading to large clusters, whereas value clusters tend to be smaller due to their inconsistent location patterns. }
Consider police complaints in Figure~\ref{fig:police-example}. This step ensures that the set of true fields appearing in the same table block, such as $B_1,B_3$, and $B_4$, are merged into one cluster, based on Proposition~\ref{prop:table}.



\noindent\textbf{Step 2: Cluster Pruning}. 
\rtwo{\sys prunes value clusters from the resulting clustering based on the observation that value clusters are typically smaller than field clusters. \sys incorporates cluster size into a confidence width metric explained shortly: larger clusters (likely fields) have lower width, indicating higher confidence, while smaller clusters (likely values) have higher width. To enhance robustness, \sys also incorporates semantic knowledge provided by LLMs. Each cluster is assigned a probability of being a field or value cluster based on the semantics of its phrases (e.g., \code{5/15/2023} is more likely a value than \code{Date}; easily judged by LLMs). } 
Let $|C|$ denote the size of cluster $C$ (i.e., number of phrases), $Pr(C)$ the probability that $C$ is a field cluster, and $w(C)$ the width of the confidence interval, respectively. Estimating $Pr(C)$ requires semantic knowledge provided by LLMs by using the prompt below, where the phrases in $C$ are used to fill the [phrases] placeholder. 
\mypython{
\small 
    LLM Prompt: Given a list of phrases labeled as either "key" or "value", identify and return the phrases that are more likely to be keys. [Phrases]}
$Pr(C)$ is estimated as the percentage of fields identified by LLMs over all phrases in $C$. $w(C)$ is computed conditioned on a confidence level of 95\% as 
$w(C) = 2\times 1.96\times \sqrt{\frac{Pr(C)(1-Pr(C))}{|C|}}$. 


\setlength{\textfloatsep}{0pt}
\begin{algorithm}[bt]
    \small
    \caption{\code{Field \update{Inference}}}
    \label{alg:field}
 \KwIn{$P$}
\textbf{/*Step 1: Phrase Clustering*/}\\
$ \mathcal{C} \leftarrow \emptyset$;  $P' \leftarrow set(P)$ \\ 
\For{$p_i \in P'$}{
$merge\_flag \leftarrow 0$ \\ 
\For{$C\in \mathcal{C}$}{
\If{$\exists p_j \in C$, s.t., $M(v_{p_i},v_{p_j})=1$}{
$C \leftarrow C\cup p_i$;
$merge\_flag \leftarrow 1$ \\ 
}
}
\If{$merge\_flag = 0$}{
$C \leftarrow \{p_i\}$, $\mathcal{C} \leftarrow \mathcal{C} \cup C$ \\
}
}
\textbf{/*Step 2: Cluster Pruning*/}\\
$G \leftarrow (V,E)$; $V \leftarrow \mathcal{C} \setminus SL$; $E \leftarrow \emptyset$ \\
\For{$(C_i,C_j), C_i,C_j \in V$}{
\If{$Pr(C_i) > Pr(C_j)$ and $w(C_i) < w(C_j)$}{
$E \leftarrow E \cup (C_i,C_j)$ \\ 
}
}
$\mathcal{C}_p \leftarrow Maximal(V)$\\
\textbf{/*Step 3: Cluster Recovery*/}\\
$F \leftarrow \mathcal{C}_p$ \\ 
\For{$C_i\in \mathcal{C}_p, C_j \in \mathcal{C}\setminus \mathcal{C}_p$}{
\If{$\exists p_i\in C_i, p_j\in C_j, L_{p_i} \leq L_{p_j}$, s.t., $PM(v_{p_i},v_{p_j})=1$}{$F \leftarrow F \cup C_j$}
}
\textbf{Return} $F$\\
\end{algorithm}

We first remove all singletons, denoted as $SL$, from the current set of clusters $\mathcal{C}$ (Line 11), as a phrase that does not match with any others is unlikely to be a field. We consider a new graph $G = (V,E)$ on clusters, where $V = \mathcal{C} \setminus SL$. 
We say cluster $C_i$ {\em dominates} $C_j$ if $Pr(C_i) > Pr(C_j)$ and $w(C_i) < w(C_j)$, implying that $C_i$ has a higher probability to be a field cluster with more confidence than $C_j$. For any pair of clusters $(C_i,C_j)$, we add an edge from $C_i$ to $C_j$ if $C_i$ dominates $C_j$ (Line 12-14). Let $\mathcal{C}_p$ be the maximal node set of $G$, where for any cluster $C\in \mathcal{C}_p$, there does not exist another cluster $C'\in V$ that  dominates $C$. After the graph is constructed by considering every pair of clusters in $V$, $\mathcal{C}_p$ is returned (Line 15).  


\noindent\textbf{Step 3: Cluster Recovery}. In Step 2, true field clusters may be pruned if a true field appears in more than one block in a record, resulting in imperfect matches with other true fields. For instance, for the invoice in Figure~\ref{fig:invoice}, the phrase \code{Start Date} appears in $B_2$, $B_3$, and $B_5$, and is not a perfect match but rather a \textit{partial perfect match} with the other unique true fields in $B_2$. Using  Proposition~\ref{prop:table}, we see that a sub-sequence of the location vector of \code{Start Date} (i.e., occurrences of \code{Start Date} in $B_2$ in every record) is a perfect match with the location vector of another  field (e.g., \code{Line}) in $B_2$. 
Here, we recover the clusters (including singletons) that are {\em partial perfect matches} with the identified clusters but weren't identified in previous steps  (Line 18-20).  
\vspace{-1mm}

%% file: template_prediction.tex
\vspace{-1mm}
\section{Template \update{Inference}}
\label{sec:template-prediction}
We now describe how \sys infers the template $T$ given the concatenated document $D$ populated with records generated by $T$. 

\subsection{Row Labeling}
\label{subsec:row-label}

Given the set of inferred fields $F$, \sys next aims to infer the structure of the template by first labeling rows. 

\subsubsection{Row Label Probabilities and Alignment}
Consider the list of phrases $P$ extracted from  $D$, and let $R = [r_1,r_2,...]$ be the set of rows in $D$ where row $r_i$ represents a list of phrases in a row. 
Each row $r\in R$ is assigned with one of the four labels, $\{\code{K}, \code{V}, \code{KV}, \code{M}\}$\footnote{While keys and fields are analogous, we use different terminologies to indicate labeling of rows and phrases respectively for clarity.}, representing $Key$, $Value$, $Key$-$Value$, and $Metadata$, respectively. In Figure~\ref{fig:template_structure_infer}, $r_4$ is a $Key$ row, $r_5$ and $r_6$ are $Value$ rows, $r_7$ is a $Key$-$Value$ row, and $r_3$ is a $Metadata$ row. 


\begin{definition}
    [Row Label Probabilities]. Consider a row $r = [p_1,p_2,...,p_n]$. Let $Prob^{r}_{K}, Prob^{r}_{V}$, $Prob^{r}_{KV}$ and  $Prob^{r}_{M}$ be the probability of the row $r$ having the label \code{K}, \code{V}, \code{KV} and \code{M}, respectively. 
\end{definition}
\vspace{-0.5em}
\vspace{-1mm}
\begin{align}
    Prob^{r}_{K} &= \frac{1}{m}\sum_{1\leq i\leq n-1}I(p_i,p_{i+1},K) \\
    Prob^{r}_{V} &= \frac{1}{m}\sum_{1\leq i\leq n-1}I(p_i,p_{i+1},V)\\
    Prob^{r}_{KV} &= \frac{1}{m}\sum_{1\leq i\leq n-1}I(p_i,p_{i+1},KV) 
\end{align}
where $m= \sum_{1\leq i \leq n-1} I(p_i,p_{i+1},K) + I(p_i,p_{i+1},V) + I(p_i,p_{i+1},KV)$


\noindent  Here $I$ is an indicator function, with $I(p_i,p_{i+1},K) = 1$ if and only if $p_i \in F$ and $p_{i+1}\in F$, denoting that the pair of phrases $(p_i,p_{i+1})$ are both fields. Similarly, $I(p_i,p_{i+1},V) = 1$ and $I(p_i,p_{i+1},KV) = 1$ denoting that the phrase pair $(p_i,p_{i+1})$ are both values (i.e., $p_i \notin F$ and $p_{i+1} \notin F$) and a key-value pair (i.e., $p_i \in F$ and $p_{i+1} \notin F$). Additionally, we set $Prob^{r}_{M} = \epsilon$,  a small positive number. ($\epsilon = 0.0001$ in our implementation.) We will explain the rationale for this probability shortly.  We then  normalize the probabilities by dividing each probability  $Prob^{r}_{x}$, $x\in \{\code{K}, \code{V}, \code{KV}, \code{M}\}$ by  $(1+\epsilon)$. Including $\epsilon$ makes a negligible impact on the probability distribution across the labels $K$, $V$, and $KV$ as $\epsilon$ is very small.


\begin{figure}[tb]
    \centering
    \includegraphics[width=1\linewidth]{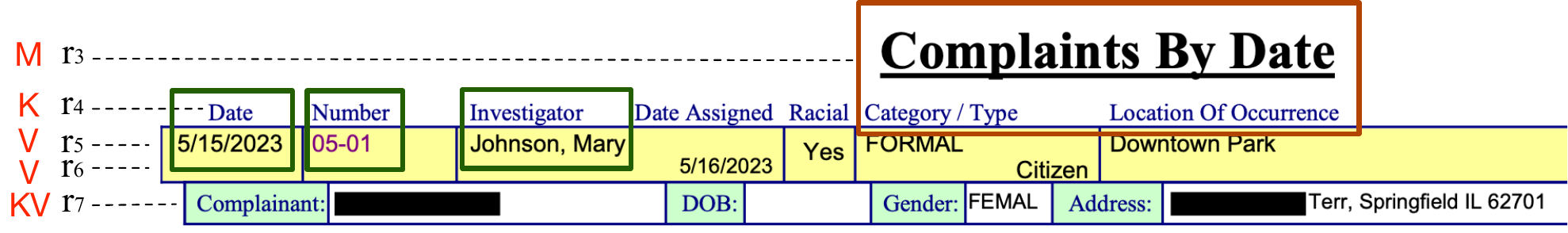}
    \vspace{-2em}
    \caption{\small Row Labels and Vertical Alignments in a Subportion of Police Records.}
    \label{fig:template_structure_infer}
\end{figure}

\begin{example}
    In Figure~\ref{fig:template_structure_infer}, consider row $r_4$ and assume all the phrases in $r_4$ are inferred as fields.  $Prob^{r_4}_{K}$ is computed by looking at every consecutive pair of phrases in $r_4$, such as (\code{Date, Number}), (\code{Number, Investigator}), (\code{Investigator, Date Assigned}). Since all phrases are inferred as fields,  $I(p_i,p_{i+1},K) = 1$ for all consecutive phrase pairs, and thus $Prob^{r_4}_{K}=\frac{1}{1+\epsilon}$, and $Prob^{r_4}_{M}=\frac{\epsilon}{1+\epsilon}$, with $Prob^{r_4}_{V}=Prob^{r_4}_{KV}=0$. Assume \code{Yes} and \code{SUSTAINED} in $r_5$ are false positives (i.e., they are values but are incorrectly inferred as fields). Then $Prob^{r_5}_{V} = \frac{3}{5\times(1+\epsilon)}$, while $Prob^{r_5}_{K} = \frac{1}{5\times(1+\epsilon)}$ and $Prob^{r_5}_{KV} = \frac{1}{5\times(1+\epsilon)}$. Note that the total number of pairs considered in the denominator in $r_5$ is different from $r_4$ since there are two missing values and two phrase pairs $(p_i,p_{i+1})$ where $p_i \notin F$ and $p_{i+1} \in F$, not considered in the label pool. 
\end{example}

Intuitively, in any label assignment, \update{{\em \bf each value row in a true table block should be vertically aligned with its header row}}. In Figure~\ref{fig:template_structure_infer}, rows $r_5$ and $r_6$ should be visually aligned with $r_4$. 
We now define the row alignment based on  bounding boxes of phrases in the row. Recall that we have defined the vertical alignment between two phrases $p_i$ and $p_j$, with bounding boxes $[x_{i1},y_{i1},x_{i2},y_{i2}]$ and $[x_{j1},y_{j1},x_{j2},y_{j2}]$ in Section~\ref{sec:background}, where $p_i$ and $p_j$ are vertically aligned, denoted as $p_i \Vdash p_j$, if $x_{j1} \leq x_{i2} \land x_{j2} \geq x_{i1}$. For example,  in Figure~\ref{fig:template_structure_infer}, \code{Date} $\Vdash$ \code{5/15/2023} since their bounding boxes overlap vertically, while \code{Date} $\nVdash$ \code{05-01}. \rmix{Note that vertical alignment based on bounding boxes is robust to minor OCR inaccuracies, as bounding boxes, if imprecise, are typically only slightly oversized and are still reliable to determine vertical alignment.} 



\begin{definition}[Row Alignment.]
\label{def:row-align}
      Two rows $r_i$ and $r_j$ are {\em well-aligned}, i.e., $A(r_i,r_j) = 1$, if $\nexists p_{j_k} \in r_j$, s.t., $p_{j_k} \Vdash p_{i_1}$ and $p_{j_k} \Vdash p_{i_2}$, $p_{i_1}, p_{i_2}\in r_i, i_1 \neq i_2$, and vice versa. Otherwise, $A(r_i,r_j) = 0$. 
\end{definition}

$A(r_i,r_j)=1$ simply implies that there does not exist a phrase in a row, say $r_i$, that overlaps with more than two phrases in the other row $r_j$.  For a Key row $r_i$ and its Value row $r_j$ in the same table, $r_i$ and $r_j$ are typically well-aligned, as no field is usually vertically aligned with more than one value. In Figure~\ref{fig:template_structure_infer}, $A(r_4,r_5)=1$ since there does not exist a phrase in $r_5$ that is aligned with two phrases in $r_4$, and vice versa. $A(r_4,r_7)=0$ since \code{Complaint} is vertically aligned with both \code{Date} and \code{Number}. 
Similarly, $A(r_3,r_4)=0$ since \code{Complaints By Date} is aligned with two phrases in $r_4$.

\subsubsection{The Row Label Assignment Problem}
We now formally state our problem of row label assignment. 
\begin{definition}[Row Labeling].
    \label{def:template-structure}
     Consider rows $R = [r_1,r_2,...]$, and inferred fields $F$. We introduce variables $y_i^{K}$, $y_i^{V}$, $y_i^{KV}$, and $y_i^{M}$ for row $r_i$, where  $y_i^{K}=1$ implies that $r_i$ has label \code{K}; otherwise, $y_i^{K}=0$. The problem of template structure inference is as follows: 
\end{definition}
\vspace{-1mm}
\begin{align}
  \text{max}\prod_{r_i\in R}  (y_i^{K}Prob_{K}^{r_i} + y_i^{V}Prob_{V}^{r_i} + y_i^{KV}&Prob_{KV}^{r_i} +  y_i^{M}Prob_{M}^{r_i})~\label{equ:goal} \\ 
  \hspace{-20mm} \text{s.t.} : 
    \forall r_i\in R, y_i^{K}, y_i^{V}, y_i^{KV}, y_i^{M} &\in \{0,1\} ~\label{equ:integer}\\
    \forall r_i\in R, y_i^{K} + y_i^{V} + y_i^{KV} + y_i^{M} &= 1 ~\label{equ:constraint} \\ 
    \forall i, y_i^{K} &\leq \sum_{j>i}A(r_i,r_j)y_j^{V} ~\label{equ:K existence}\\ 
    \forall i, y_i^V &\leq \sum_{j<i}A(r_j,r_i)y_j^K ~\label{equ:Vexistence} 
\end{align}

The row labeling problem in  Definition~\ref{def:template-structure} aims to find the {\em most probable row label assignments} by maximizing the products of the probabilities from all rows in $R$ as in (\ref{equ:goal}), 
under various constraints. Constraints~(\ref{equ:integer}) and (\ref{equ:constraint}) ensure that each row is assigned exactly one label. 
Constraint (\ref{equ:K existence})  states that for each \textit{Key} row $r_i$, there must exist a \textit{Value} row $r_j$ under $r_i$ ($j>i$) aligned with it ($A(r_i,r_j)=1$). Similarly, Constraint (\ref{equ:Vexistence}) says that for each \textit{Value} row $r_j$, we expect a \textit{Key} row $r_i$ before $r_j$ ($i<j$) with $A(r_i,r_j)=1$, i.e., the values are aligned with keys.   





If the \textit{Metadata} label is not introduced, a row $r_i$ could be assigned a label with $0$ probability if it violates Constraints (\ref{equ:K existence}) and (\ref{equ:Vexistence}). For example, if $r_i$ has $Prob_K^{r_i} = 1$ but lacks a corresponding $Value$ row beneath it, it would be assigned a label with 0 probability, leading to a poor assignment. To prevent this, we introduce the \textit{Metadata} label with a low $\epsilon$ probability, ensuring the worst possible label for a row is \textit{Metadata}. This adjustment has minimal impact on other rows, as a \textit{Metadata} row does not affect the template structure (e.g., inferring table or key-value blocks), nor does it interfere with the assignment of labels with non-zero probability. Empirical observations confirm the effectiveness of the \textit{Metadata} label. True \textit{Metadata} rows are often misclassified as \textit{Key} rows (e.g., $r_3$ in police complaints) or \textit{Value} rows (e.g., $r_{9}$ in invoices). Such misclassified rows are typically {\em misaligned} with true \textit{Key} rows, allowing them to be easily identified.

Since the objective in (\ref{equ:goal}) is non-linear, we convert it to be linear below for ease of optimization. 

\begin{theorem}
\label{theo:problem}
    The following is equivalent to Eq.~(\ref{equ:goal}), \\ 
    \vspace{-2mm}
    \begin{align}\label{equ:obj}
        \max \sum_{r_i \in R} (& y_i^{K}\text{log}(Prob_{K}^{r_i}) + y_i^{V}\text{log}(Prob_{V}^{r_i}) + y_i^{KV}\text{log}(Prob_{KV}^{r_i}) + y_i^{M}\text{log}(Prob_{M}^{r_i})) 
    \end{align}
\end{theorem}

We show the proof of Theorem~\ref{theo:problem} in Appendix~\ref{subsec:proof-theorem1}.

\subsubsection{Solving Row Label Assignment: ILP and Hardness}
After linearization, the problem defined in Definition~\ref{def:template-structure} with the new objective in (\ref{equ:obj}) is an Integer Linear Programming (ILP) problem.  \techreport{The function $A(\cdot)$  can be precomputed for all row pairs ($r_i,r_j$) before solving the ILP problem, which makes all the constraints fixed and linear.} We additionally can show the following: 


\begin{theorem}
\label{theo:np-hard}
    {\sc Row labeling} is {\sc NP-hard}. 
\end{theorem}

\papertext{
\noindent We prove Theorem~\ref{theo:np-hard} via a reduction from Vertex Cover in Appendix~\ref{subsec:proof-np-hard}. }
\techreport{We prove Theorem~\ref{theo:np-hard} via a reduction from Minimum Vertex Cover  below.}
\techreport{\input{npproof}
}
Unfortunately, solving the above ILP  can be expensive.

\topic{Pruning Strategy} Thankfully, inferring the template does not require all of $D$. Instead, we can use a consecutive subsequence, $r_i,r_{i+1},...,r_{i+j}$, of rows $R$ containing {\em at least one record}  from the true template. Let $R' \subseteq R$ be the smallest consecutive subsequence of rows in $R$ such that each inferred field $p \in F$ appears at least twice in $R'$. $R'$ is obtained by scanning rows of $R$ in ascending order and stopping once each inferred field appears twice. We then have: 

\begin{theorem}
\label{theo:input}
    Let $F'$ be the true fields in the true template $T'$, and $F$ the inferred fields. If $\exists p\in F \cap F'$, and $\forall Rec$, $p$ appears exactly once in $Rec$, then $R'$ contains at least one record $Rec$ created by $T'$. 
\end{theorem}

\noindent Theorem~\ref{theo:input} guarantees that as long as there is a field $p$ that is correctly inferred and  appears exactly once in every record, then $R'$ contains at least one complete record. 
\papertext{See proof in Appendix~\ref{subsec:proof-theorem3}}.
\techreport{
\begin{proof}[Proof of Theorem~\ref{theo:input}]~\label{proof:input}
If there exist a correctly inferred field $p\in F$ that appears exactly once in every record, and $R'$ contains every field at least twice, then $R'$ contains at least one complete record given that $R'$ is a consecutive sublist of rows. 
\end{proof}
}

\sys considers $R'$ as input to the row labeling problem, and further uses the output row labels to infer the template in Section~\ref{subsec:template_infer}.  
Empirically, thanks to the small size of input rows $|R'|$ (roughly 32 rows on average in our benchmark), an ILP solver~\cite{Gurobi} provides a solution in milliseconds.

\subsection{Template Inference}
\label{subsec:template_infer}

Given $R' = [r_1, r_2, \dots, r_n]$ with row label inferences, we aim to determine the structure of the underlying template. Let $B.phrases$ be the set of phrases in a data block $B$. 

We begin by initializing an empty template tree $T = (V, E)$ with an artificial root. 
1) For every row $r_i$ labeled \code{K}, we create a tree node $v_i$ with $v_i.type = table$ and $v_i.fields = r_i$. 
2) We merge together as many consecutive rows labeled \code{KV} as possible into a block $B$, and create the corresponding node $v_j$ where $v_j.type =$ {\em Key-Value} and $v_j.fields = F \cap B.phrases$.   
3) For each row $r_i$ labeled as \code{V}, we assign it to its \textit{closest aligned key row} preceding $r_i$, denoted as $r_j$, where $j < i$, $r_j.label = \code{K}$, $A(r_j, r_i) = 1$, and $\nexists r_k, j < k < i$, such that $r_k.label = \code{K}$ and $A(r_k, r_i) = 1$. 
Note that if a newly created node $v_i$ has the same $v_i.type$ and $v_i.fields$ as an existing node $v_j$ in $T$, $v_i$ will not be inserted into $T$. 
Below, we use two examples to illustrate node creation and tree assembly.  

\begin{figure}[tb]
    \centering
    \includegraphics[width=0.7\linewidth]{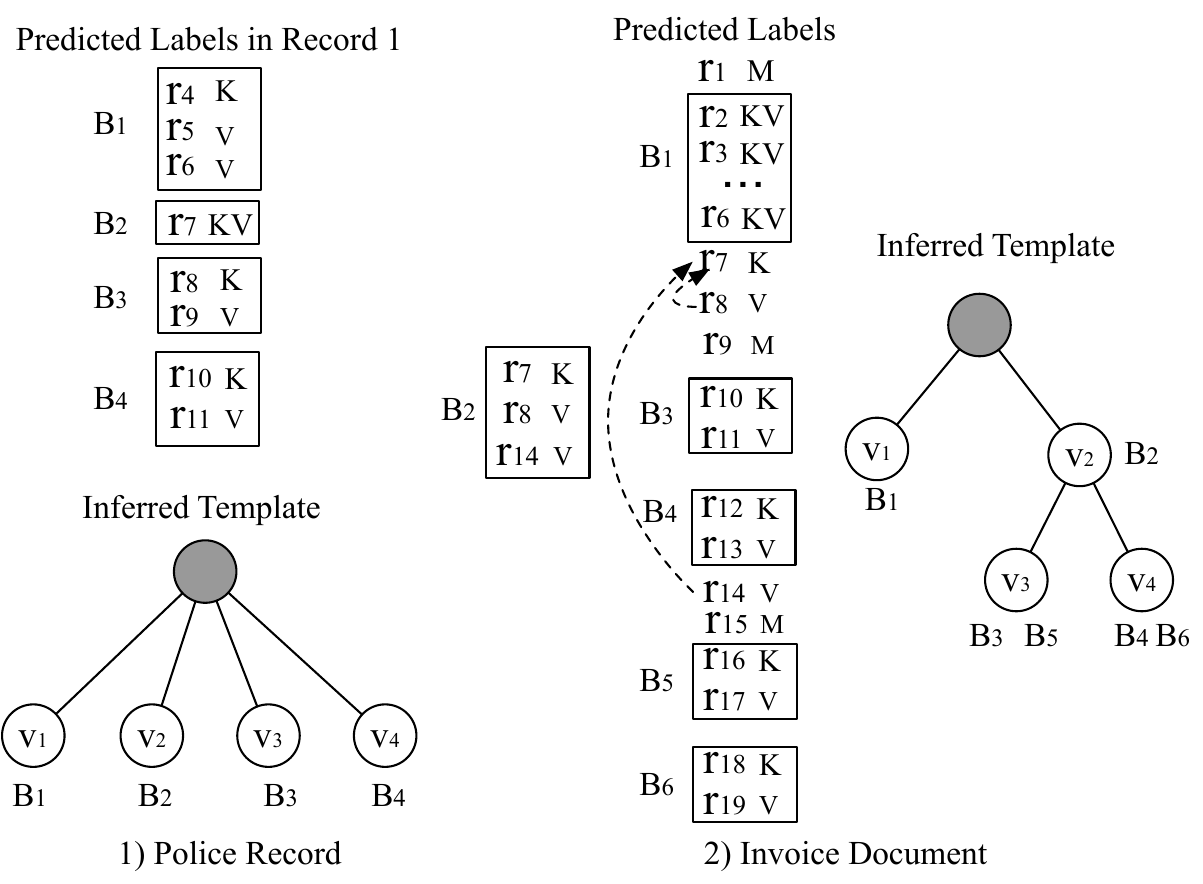}
    \vspace{-1em}
    \caption{\small Template Structure Inference. }
    \label{fig:template_infer} 
\end{figure}

\begin{example}
In Figure~\ref{fig:template_infer}, inferred row labels are shown next to the rows. Rows within the same table, such as $r_1$ and $r_2$ in Figure~\ref{fig:template_infer}-1, are aligned and presented together in a block. Rows grouped in the same table block are well-aligned according to  Definition~\ref{def:row-align}. 
Consider police complaints in Figure~\ref{fig:template_infer}-1. Rows labeled \code{K}, specifically $r_4$, $r_8$, and $r_{10}$, trigger creation of nodes $v_1$, $v_3$, and $v_4$, respectively, while $r_7$, labeled \code{KV}, corresponds to $v_2$. The blocks are sequentially placed with no overlap, with all nodes as leaves. The pre-order traversal of the corresponding nodes in the tree is sorted by block indices (e.g., $loc(B_1) < loc(B_2)$), resulting in the template $T$ depicted in Figure~\ref{fig:template_infer}-1. 
In the invoice in Figure~\ref{fig:template_infer}-2, consecutive \code{KV} rows ($r_2$ to $r_6$) are merged into a single \code{KV} block $B_1$. In table block $B_2$, $r_7$ is the closest key row aligned with $r_{14}$, making $r_{14}$ a value row in $B_2$. $B_3$ and $B_5$ correspond to the same node $v_3$ since $r_{17} = v_3.fields$. Therefore, $v_3$ is created only once, triggered by $r_{10}$. 
Since data block $B_2$ (corresponding to $v_2$) overlaps with $B_3$ and $B_4$ (corresponding to $v_3$ and $v_4$), $v_2$ becomes the parent of $v_3$ and $v_4$, while $v_1$ is a leaf node since its block doesn't overlap with others.
\end{example}

The complexity of template inference based on inferred row labels is $O(|R'|^2)$, which is efficient since $|R'|$ is typically small.

%% file: npproof.tex
\begin{proof}[Proof of Theorem~\ref{theo:np-hard}]~\label{proof:np-hard}
We reduce the \textsc{Minimum Vertex Cover} problem to an instance of the row labeling problem. Consider an instance of vertex cover problem, where we are given an undirected graph \(G=(V,E)\), with
$|V| = n, |E|=m,$ and a parameter $k$. The decision version of \textsc{Minimum Vertex Cover} asks: Is there a subset $S\subseteq V,\ |S|\le k,$ such that every edge has at least one endpoint in $S$?  

We now consider an instance of the row labeling problem. We create:
$
\underbrace{v_1,\dots,v_n}_{\text{vertex rows}}
\quad\cup\quad
\underbrace{\,e_1,\dots,e_m\,}_{\text{edge rows}}
$
where each \(v_i\) corresponds to a vertex in \(V\) and each \(e_k\) corresponds to an edge in \(E\). We index these rows so that
$
\mathrm{Index}(v_i) = i \quad (i=1,\dots,n),
\quad
\mathrm{Index}(e_k) = n + k \quad (k=1,\dots,m).
$
Hence any vertex row \(v_i\) has index \(i\le n\), 
and any edge row \(e_k\) has index \(n+k > n\). 

We also introduce ``columns'' of phrases per row to ensure the right alignment across the rows. We have a distinct ``column''  for each \emph{edge row}. For example, if $(a,b)$, $(a,c)$, $(b,c)$, and $(b,d)$ are edges, we introduce $c_{ab}$, $c_{ac}$, $c_{bc}$, and $c_{bd}$ in some order. Such an order of columns must be consistent across the rows but can be arbitrary. That column aligns with exactly those vertex rows that are the endpoints of that edge (e.g.\ the column for edge \(ab\) aligns only with rows \(a\) and~\(b\)). In matrix terms, this means:
$A(a,\,ab)=1, A(b,\,ab)=1,\text{and }A(x,\,ab)=0 \text{ for any other row }x.$ 

To ensure this constraint, we introduce text per column. Suppose the columns are indexed $1,2,3,\dots,m$. Text phrases for column $i$ begin at $\alpha_i$ and extend until $\alpha_i + \beta$, where $\alpha$ and $\beta$ are integers greater than 2, and $\alpha > \beta$. The gap between $\alpha_i + \beta$ to $\alpha(i+1)$ ensures a correct phrase extraction. We then add phrases between $\alpha_i$ and  $\alpha_i + \beta$ as follows. For a column corresponding to $e_i = (v_x,v_y)$, that we call $c_{v_x,v_y}$. 
\begin{itemize}
    \item For edge row $e_i$, we add a phrase from $[\alpha_i + \beta-1, \alpha_i + \beta)$. 
    \item For vertex $v_x$ and $v_y$, we add a phrase from $[\alpha_i, \alpha_i + \beta)$. 
    \item For all other rows, we add a phrase from $[\alpha_i, \alpha_i + \beta-2]$. 
\end{itemize}

Thus, an edge row $e_i = (v_x,v_y)$ is only aligned with its vertex rows $v_x$ and $v_y$, and not any other edge or rows. Aggregating information across all rows, we have $A(e,v)=1$ iff $v$ is one of the endpoints of $e$. $A(e_i,e_j) = 0, \forall i \neq j$. $A(v_i,v_j)=1, \forall i \neq j$. 

In the following, we will set probabilities such that all edge rows are set as label \code{V}, a subset of vertex rows corresponding to a vertex cover will be set as label \code{K}, and the remaining vertex rows are set as label \code{KV}. In the instance of row labeling problem, each row $r_i$ has four binary variables:
$y_i^K,\; y_i^V,\; y_i^{KV},y_i^M$ with $y_i^K + y_i^V + y_i^{KV} + y_i^M = 1$.  Constraint~\ref{equ:K existence},  $y_i^{K} \leq \sum_{j>i}A(r_i,r_j)y_j^{V}$, enforces that if a row is labeled \code{K}, it must find a neighbor of higher index labeled \code{V}.  
In our construction, any vertex row $v_i$ (index $i\le n$) that is \code{K} must see some edge row $e_k$ (index $n+k>i$) labeled $V$ for which $A(v_i,e_k)=1$ (meaning $v_i$ is an endpoint of $e_k$). Similarly, Constraint~\ref{equ:Vexistence},  $y_i^V \leq \sum_{j<i}A(r_j,r_i)y_j^K$, enforces that if a row is labeled \code{V}, it must find a neighbor of lower index labeled \code{K}. 
In our construction, any edge row $e_k$ (index $n+k$) that is \code{V} must see some vertex row $v_i$ (index $i<n+k$) labeled \code{K} with $A(v_i,e_k)=1$. 
So each edge is covered by a \code{K} from one of its endpoints. 

We now assign \emph{probabilities} for each row+label combination, $Prob^{r_i}_{\ell}$, where $\ell\in\{K,V,KV,M\}$.  
We will transform the objective in Eq.~(\ref{equ:obj})  into a \emph{decision} version by imposing a threshold $\Theta$ and asking if $\sum_{r_i \in R} (y_i^{K}\text{log}(Prob_{K}^{r_i}) + y_i^{V}\text{log}(Prob_{V}^{r_i}) + y_i^{KV}\text{log}(Prob_{KV}^{r_i}) +  y_i^{M}\text{log}(Prob_{M}^{r_i})) \geq \Theta$. 
We set $Prob^{r_i}_{M} = 0$ for any row $r_i$, hence no solution will assign $M$ to any row. We can change the proof to admit $M$ at a small probability, but we omit it for simplicity. For each edge row $e_k$, we set $Prob^{e_k}_{V}=1$, $Prob^{e_k}_{K}=0$, and    $Prob^{e_k}_{KV}=0$. Thus each edge row is effectively forced to be \code{V}. For each vertex row $v_i$, we set $Prob^{v_i}_{K} = \lambda$, $Prob^{v_i}_{KV} = \mu$, and $Prob^{v_i}_{V} = 0$, where $\mu > \lambda$. Hence, labeling $v_i$ as $K$ consitrubtes $log(\lambda)$ to the sum, while labeling $v_i$ as $KV$ consitrubtes $log(\mu) > log(\lambda)$.  Thus, if we label $x$ vertices out of $n$ as $K$, we add $log(\lambda)x + log(\mu)(n-x)$. Let $k$ be the parameter in the decision version of the vertex cover problem. Let $\Theta = log(\lambda)x + log(\mu)(n-x)$. And therefore the decision version of our row labeling problem becomes  $\sum_{r_i \in R} \sum_{l\in\{K,V,KV,M\}}y_i^{l}\text{log}(Prob_{l}^{r_i}) \geq \Theta$. 

We now show the correctness of the above reduction. On one hand,  ($\Rightarrow$) If $G$ has a vertex cover $S$, $|S|\le k$. Label the vertex rows in $S$ as $K$, the other vertex rows as $KV$, and all edge rows as $V$.  
\begin{itemize}
\item \emph{All alignment constraints are satisfied:}
  - Each edge row labeled $V$ sees a $K$ neighbor above from among its endpoints in $S$.
  - Each $K$ vertex row sees at least one incident edge row $V$ below; otherwise it will not be part of the cover. 
\item \emph{Total log-prob $\ge\Theta$:}
  - We used $\le k$ vertices as $K$, so \\$\sum_{r_i \in R} \sum_{l\in\{K,V,KV,M\}}y_i^{l}\text{log}(Prob_{l}^{r_i}) \geq \Theta$. 
\end{itemize}

On the other hand, ($\Leftarrow$) if the row labeling problem meets alignment requirement and the objective is greater than $\Theta$, 
\begin{itemize}
    \item No row can be labeled $M$.
    \item Every edge row is $V$. Labeling it otherwise means we can improve the objective by changing it to be so. 
    \item At most $k$ vertex rows can be labeled $K$. If we had more than $k$ vertex rows labeled as $K$, we will have an objective less than $\Theta$. 
    \item Because each edge row $e_k$ is $V$, it must see a $K$ neighbor above.  So the set of $K$-labeled vertices forms a valid vertex cover of size $\le k$. 
\end{itemize}
Hence there is a size-$\le k$ vertex cover in $G$ if and only if there is a labeling with alignment constraints satisfied and total log-prob $\ge\Theta$.  Since \textsc{Minimum Vertex Cover} is NP-complete, 
deciding whether 
\[
\max \;\sum_{r}\sum_{\ell\in\{K,V,KV,M\}} y_r^\ell\log(\mathrm{Prob}^{r}_{\ell}) \;\ge\;\Theta
\]
is NP-hard. Thus the row labeling problem is also NP-hard.
\end{proof}  
Under the well-known assumption that $NP\neq P$, there does not exist a polynomial-time solution that can solve the problem optimally.

%% file: extraction.tex
\section{Data Extraction}
\label{sec:data-extraction}


Given the inferred template $T$ and the concatenated document $D$, \sys aims to extract data from the set of records $\mathcal{R}ec = [Rec_1, Rec_2, \dots]$ generated by $T$ within $D$. To do so, we first separate the given document $D$ to into records based on  $T$ (Section~\ref{subsec:record_separation}). Within each record, we further identify the data blocks (Section~\ref{subsec:block-seperation}), and then extract data from table and key-value blocks (Section~\ref{subsec:data_extraction}). 

\subsection{Record Separation}
\label{subsec:record_separation}




Consider the row representation $R$ of $D$, $R = [r_1, r_2, \dots]$. Recall that a block $B$ is a sequence of rows.   
For any node $v \in V$ in the template $T$, we say that node $v$ is \textit{visited} by block $B$ if all the fields of $v$ appear in block $B$, i.e., $v.fields \subseteq B.phrases$, denoted by $vis(v, B) = \text{True}$. Intuitively, a record is the smallest block that visits \textit{every} node in $T$ \textit{at least once} via a \textit{pre-order} traversal. The document $D$ is now separated into a list of records $\mathcal{R}ec$.


\begin{figure}[tb]
    \centering
    \includegraphics[width=0.3\linewidth]{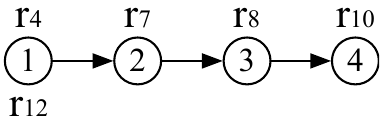}
    \vspace{-1em}
    \caption{\small Record Separation in Police Complaints.}
    \label{fig:record}
\end{figure}

\begin{example} 
Given the inferred template $T$  in Figure~\ref{fig:template_infer}-1, we aim to separate records 1 and 2 in Figure~\ref{fig:police-example}.  We scan each row $r_i\in D$ in ascending order of location, and consider the pre-order traversal of nodes in $T$ in Figure~\ref{fig:record}. $vis(v_1,[r_4])$ is True as $v_1.fields = r_4$ ($r_4$ is the header of the table block corresponding to $v_1$), and $v_1$ is thus visited. Similarly, $[r_7]$ visits $v_2$ as $v_2.fields \subseteq r_7$, and $[r_{10}]$ visits $v_4$. When $r_{12}$ visits $v_1$ the second time ($v_1.fields = r_{12}$), all nodes in $T$ have been visited at least once, and $v_1$ being visited by $r_{12}$ the second time implies the start of a new record. Note that it is not necessary that each node is visited exactly once, because nested nodes (e.g., $v_3$ in invoice template) may be visited multiple times. 
\end{example}



\subsection{Block Separation}
\label{subsec:block-seperation}
Given a record $Rec\in \mathcal{R}ec$ from the previous step, we aim to identify a list of blocks corresponding to each node in $T$. Figures~\ref{fig:police-example} and~\ref{fig:invoice} show the data blocks within a record in  police complaints and  invoices, respectively.

Block separation proceeds in two steps as in Algorithm~\ref{alg:block}. First, we assign the labels for each row $r\in Rec$ based on the template $T$. Given a row $r = [p_i,\dots,p_j]$, if $\exists v\in V, v.type=$\textit{Table}, $\forall p\in v.fields, p\in r$, then $r.label = K$ (Lines 2-4).  Otherwise, if $\exists v\in V, v.type = ${\em Key-Value}, $r \cap v.field \neq \emptyset$, $r.label=$ $KV$ (Lines 5-6). In all the other cases, $r.label = V$ (Lines 7-8). 
Second, for a value row $r_i$ whose label is $V$, let $r_j$ be the closest key row preceding $r_i$. If $r_i$ is aligned with $r_j$, i.e., $A(r_j,r_i)=1$, then $r_i$ is assigned to be a value row of the table with the key row as $r_j$ (Lines 11-14). Otherwise, if there exists another key row $r_l$ preceding and as close as possible to $r_i$, where $A(r_l,r_i)=1$ and the corresponding node $v$ is not a leaf, then $r_i$ is assigned to be a value row of $r_l$ (Lines 15-18). 
In  other cases, $r_i$ is labeled as \textit{Metadata} (Lines 19-20), as it is neither aligned with any $K$ row (and thus not a valid $V$ row) nor inferred to be a $KV$ or $K$ row. Finally, for any consecutive key-value rows, we merge them together into one block (Lines 21-23). 

\begin{example}
Consider the police complaints in Figure~\ref{fig:police-example} and its inferred template in Figure~\ref{fig:template_infer}-1. The labels for $r_4$, $r_8$, and $r_{10}$ are \code{K}, $r_7$'s label is \code{KV}, and all other rows are labeled \code{V}. $r_5$ is assigned to $r_4$ because $r_4$ is aligned with $r_5$ and is the closest preceding key row. Similarly, $r_9$ is assigned to $r_8$. 
Now consider the invoice document in Figure~\ref{fig:invoice} and its inferred template in Figure~\ref{fig:template_infer}-2. The closest preceding key row to $r_{14}$ is $r_{12}$. However, $r_{12}$ is not the key row for $r_{14}$ because they are not aligned. The closest preceding key row aligned with $r_{14}$ is $r_7$, corresponding to node $v_2$ in the template. Since $v_2$ is not a leaf node, its data block can overlap with others as defined in Section~\ref{sec:template}. Thus, $r_{14}$ is correctly assigned to $r_7$.
\end{example}


Each record $Rec \in \mathcal{R}ec$  is now transformed into a list of data blocks, either table or key-value blocks. 



\setlength{\textfloatsep}{0pt}
\begin{algorithm}[bt]
    \small
    \caption{\code{Block\_Separation}($Rec,T=(V,E)$)}
    \label{alg:block}
/*\textbf{Step 1: Row Label Assignment}*/\\
\For{$r \in Rec$}{
\If{$\exists v\in V, v.type=$table$, v.fields = r$}{
$r.label$ = K\\
}
\uElseIf{$\exists v\in V, v.type=$key-value$, r \cap v.field \neq \emptyset$} {
$r.label$ = KV\\
}
\Else{
$r.label=$V\\
}
}
/*\textbf{Step 2: Block Separation}*/\\
$Member \leftarrow \emptyset$ \\ 
\For{$r_i\in Rec, r_i.label=$V}{
$r_j \leftarrow $ \code{closest\_preceding\_key\_row}($r_i$) \\ 
\eIf{$A(r_i,r_j)=1$}{
    $Member(r_j).append(r_i)$ \\  
}{
$r_l \leftarrow$  \code{closest\_preceding\_aligned\_key\_row}($r_i$) \\ 
\eIf{$\exists r_l$, $v_l.hasChild=True$}{
$Member(r_l).append(r_i)$ \\ 
}{
$r_i.label$ = Metadata
}
}
}
\For{$r_i,r_{i+1}\in Rec$}{
\If{$r_i.label=$KV and $r_{i+1}.label=$KV}{
$Member(r_i).append(r_{i+1})$\\
}
}
\textbf{Return} $Member$ \\ 
\end{algorithm}





\subsection{Data Extraction}
\label{subsec:data_extraction}

We then extract data from the table and key-value blocks.

\noindent\textbf{Data Extraction in Table Blocks}. Given a table block $B$, we next extract its contents. Let $B.fields=[f_1, f_2, \dots, f_n]$ be the list of inferred fields in $B$ sorted in ascending order of location. 

The first row $r_l$ in $B$ corresponds to $B.fields$, and each value row $r_i \in B, r_i \neq r_l$ corresponds to a tuple of this table, where $r_i = [p_{i1}, p_{i2}, \dots, p_{im}]$. Intuitively, values corresponding to a field are typically aligned vertically as a column in a table block. Thus, given a value row $r_i\in B$, if $p_{ij}$ is vertically aligned with $f_j$, i.e., $p_{ij} \vDash f_j$, then $p_{ij}$ is determined to be a value of $f_j$. Given a field $f_j$, if there does not exist a value $p_{ij} \in r$, such that $p_{ij} \vDash f_j$, then $f_j$'s value is {\em missing} (or {\em NULL}) for row $r_i$. Consider data extraction of block $B_1$ in Record 1 for police complaints in Figure~\ref{fig:police-example}. The value for \code{Completed} is missing in $r_5$, \code{05-01} and \code{Yes} are the values for \code{Number} and \code{Recorded On Camera} in $r_5$, respectively.

\setlength{\textfloatsep}{0pt}
\begin{algorithm}[bt]
    \small
    \caption{\code{KV-Extract}($B$)}
    \label{alg:kv}
$seen \leftarrow \{\}$; $i\leftarrow 0$; $KV \leftarrow \emptyset$ \\ 
\For{$p_i \in B$}{$seen[i] \leftarrow \code{False}$}
\While{$i<|B|-1$}{
    \If{$seen[i] == $ \code{False}}{
        \If{$p_i \in B.fields \land p_{i+1} \notin B.fields$}{
            $KV \leftarrow KV \cup (p_i,p_{i+1})$ \\ 
            $seen[i+1] \leftarrow \code{True}$
        }
        \If{$p_i \in B.fields \land p_{i+1} \in B.fields$}{
            $KV \leftarrow KV \cup (p_i,missing)$ \\ 
        }
    }
    $i \leftarrow i+1$\\ 
}
\textbf{Return} $KV$
\end{algorithm}

\noindent\textbf{Data Extraction in Key-Value Blocks}. In a key-value block $B$, a field either has a corresponding value or may be missing. In Algorithm~\ref{alg:kv}, to extract key-value pairs given inferred fields, we sequentially scan each consecutive phrase pair $(p_i, p_{i+1})$ in $B$. We use an array $seen[i]$ to track whether the phrase $p_i$ has been scanned. If $(p_i, p_{i+1})$ forms a key-value pair, i.e., $p_i \in B.fields \land p_{i+1} \notin B.fields$, it is added to the extraction result (Lines 6-7).  We then examine the next phrase pair starting from $p_{i+2}$ by setting $seen[i+1]$ to \texttt{True}. If both phrases are fields, the phrase $p_i$ is assigned a {\em missing} value, and we proceed by examining the phrase pair starting with $p_{i+1}$ (Lines 9-10). For example, if \code{Complaint} and \code{DOB} in row $r_7$ in Figure~\ref{fig:police-example} are identified as fields, their values are inferred as \code{missing}.  

\update{Finally, we assemble the extracted data from each block in every record into a tree, where we describe the representation of the data extraction results in Appendix~\ref{susbec:data-extraction-objects}. }




\topic{\update{Metadata Preservation}} 
In addition to extracting structured data, \sys stores metadata and preserves its associations with the extracted data.  To achieve this, \sys stores the bounding box and page number for every metadata phrase. 
\rone{In police complaints in Figure~\ref{fig:police-example}, the phrase \code{Complaint \#1} in row $r_{9}$ is extracted as metadata since it is neither a table header nor a value (i.e., it lacks an associated column header). Here, \sys extracts the table with schema $\{$\code{Type of Complaint, Description, Complaint Disposition}$\}$ while preserving spatial relationships between metadata and table content using bounding boxes. For example, the preserved bounding boxes suggest that \code{Complaint \#1} is horizontally aligned with the first row of the extracted table. Users can leverage this alignment to expand the extracted table  by adding a new column called \code{Complaint} and populating the first row with  value \code{1} for it. Since the interpretation and transformation of metadata depends on use cases, \sys leaves this task to end users. However, \sys preserves the spatial relationships between metadata and extracted data to enable such transformations. }  


\topic{\update{Correctness of End-to-end Approach}} 
\papertext{\update{Finally, we establish the correctness of the end-to-end \sys pipeline in one commonly observed type of documents, that we call {\em compliant documents}. We leave the definition of compliant documents and the correctness of \sys in Appendix~\ref{subsec:appendix-correctness}.}}

%% file: experiment.tex
\vspace{-1mm}
\section{Experimental Evaluation}
\label{sec:exp}

\vspace{-1mm}
We evaluate the performance of \sys on \update{two benchmarks comprising 64 real-world templatized document collections}. 

\vspace{-1mm}
\subsection{Evaluation Setup}
\label{subsec:exp-setup}
\noindent\textbf{Datasets}. 
\update{We constructed two benchmarks for our evaluation. The first benchmark, {\em Q-Benchmark}, focuses on evaluating the quality of extraction results. The second benchmark, {\em S-Benchmark}, examines the scalability of the compared tools.
 } 
\update{Q-Benchmark consists of }
 34 real-world datasets from our collaborators,
and three open benchmarks~\cite{naparstek2024kvp10k,xu2022xfund,vsimsa2023docile}, spanning diverse domains such as police use of force documents, invoices, grant reports, order bills, certification records, contracts, and trade forms. 
We randomly sampled 5 to 30 documents per dataset (\update{around 7.4 documents and 7189 tokens on average}), based on dataset size, with more samples for smaller datasets and fewer for larger ones to ensure representative coverage. The ground truth for all datasets was manually collected and verified by three human labelers over a month, highlighting the task's complexity even for humans. \update{The S-Benchmark consists of 30 datasets collected from the open benchmark~\cite{naparstek2024kvp10k}, containing around 2,133 documents and over two million tokens per dataset on average. Given its scale, we do not have ground truth for S-Benchmark, so we primarily focus on evaluating latency and cost.  }

We classify datasets per benchmark into three types based on complexity: \code{Easy}, \code{Medium}, and \code{Hard}. \code{Easy} datasets have templates with only one node besides the virtual root, meaning each document contains either a pure table or key-value block. For simplicity, we omit the virtual root when discussing nodes in templates.  
\code{Medium} datasets feature templates with more than two nodes, all in the same layer (i.e., all nodes are leaves), indicating sequentially placed, non-overlapping data blocks within a record.  
\code{Hard} datasets have nodes with children, implying overlapping data blocks, as in invoices in Figure~\ref{fig:invoice}.  
Table~\ref{tab:datasets} summarizes the characteristics of the datasets. 





\noindent\textbf{Tools Compared}. We compare \sys with six baselines. This includes Amazon $Textract$~\cite{textract}, and Azure AI Document Intelligence (\textit{AzureDI} for short)~\cite{azure}, which detect and extract tables and key-value pairs from PDFs. We also include two state-of-the-art vision-based LLMs, $vLLM$-$S$ and $vLLM$-$C$, both using \code{gpt4vision}~\cite{gpt4vision} from OpenAI. $vLLM$-$S$ uses a prompt to identify the template structure and then extracts data, while $vLLM$-$C$ directly extracts all key-value pairs from each page. For a table, $vLLM$-$C$ extracts each cell and its header as key-value pairs, whereas for key-value blocks, it outputs a list of key-value pairs. \rone{We also considered LLMs based on OCR text, where we provide hints that records are generated from the same template and append multiple example records. As vision-based LLMs outperform any text-based LLMs we tried, we omit them below. The prompts for LLMs baselines are in Appendix~\ref{subsec:appendix-prompts}. }
\update{Finally, we compare \sys with {\em Evaporate}~\cite{arora2023language}, a tool that extracts structured data from documents using LLMs. Evaporate operates on OCR (text) output from documents.  Evaporate has two versions: {\em Evaporate-Direct} ({\em Eva-D} for short), which uses LLMs to directly identify fields and extract values, and {\em Evaporate-Code} ({\em Eva-C} for short), which uses LLMs to generate extraction scripts and then applies those scripts to extract values.  } 
\rtwo{\sys uses Pdfplumber~\cite{pdfplumber} as the OCR tool, which can be replaced by alternatives such as Tesseract~\cite{tesseract} or MistralOCR~\cite{mistral}, which extract phrases along with bounding boxes. }

\begin{table}[tb]
\vspace{-2mm}
\small
    \centering
    \begin{adjustbox}{width=0.6\textwidth,center}
    \begin{tabular}{|c||c|c|c|}
    \hline
    \multicolumn{4}{|c|}{\textbf{Q-Benchmark} } \\ \hline
    & \textbf{\# of Datasets} & \textbf{Avg Dataset Size (tokens)} & \textbf{Avg  \# of Doc}   \\ \hline
    \code{Easy} & 6 & 5902  & 11.2 \\ \hline 
    \code{Medium} & 21 & 7813 & 6.8 \\ \hline 
    \code{Hard} & 7 & 6417 & 5.1 \\ \hline 
    \multicolumn{4}{|c|}{\textbf{S-Benchmark} } \\ \hline
    \code{Easy} & 23 & 2, 106, 218  & 2117 \\ \hline 
    \code{Medium} & 6 & 2, 070, 688 & 2191 \\ \hline 
    \code{Hard} & 1 & 1, 059, 775 & 1577 \\ \hline 
    \end{tabular}
    \end{adjustbox}
 \caption{\small Characteristics of 64 Datasets.} 
    \vspace{-2em}
    \label{tab:datasets}
\end{table}

\begin{table*}[]
\begin{adjustbox}{width=1\textwidth,center}
\begin{tabular}{|c|c| c| c|c|c|c|c||c| c| c| c|c|c|c|}
\hline
\multirow{2}{*}{} & \multicolumn{7}{c||}{\textbf{Precision} } & \multicolumn{7}{c|}{\textbf{Recall} }  \\  \cline{2-15}
& Textract & vLLM-S & vLLM-C & AzureDI & Eva-D & Eva-C &  \sys & Textract & vLLM-S & vLLM-C & AzureDI & Eva-D & Eva-C & \sys \\ \hline
\code{Easy} &	0.9 & 0.74 & 0.73 & 0.54 & \revised{0.47} & 0.32 & \textbf{0.96} & 0.88 & 0.62 & 0.68 & 0.68 & \revised{0.52} & 0.25 & \textbf{0.98} \\ \hline
\code{Medium} & 0.5 & 0.63 & 0.6 & 0.49 & \revised{0.41} & 0.33 & \textbf{0.89} & 0.51 & 0.57 & 0.5 & 0.62 & \revised{0.33} & 0.26 & \textbf{0.9} \\ \hline
\code{Hard} & 0.38 & 0.59 & 0.64 & 0.35 & \revised{0.45} & 0.35 & \textbf{0.85} & 0.44 & 0.49 & 0.57 & 0.66 & \revised{0.44} & 0.26 & \textbf{0.91} \\ \hline
\end{tabular}
\end{adjustbox}
\caption{\small Precision and Recall on \code{Easy}, \code{Medium} and \code{Hard} Datasets on Q-Benchmark. }
\label{table:quality}
\vspace{-1em}
\end{table*}

\begin{table*}[]
\vspace{-15pt}
\begin{adjustbox}{width=1\textwidth,center}
\begin{tabular}{|c|c| c| c|c|c|c|c||c| c| c| c|c|c|c|}
\hline
\multirow{2}{*}{} & \multicolumn{7}{c||}{\textbf{S-Precision} } & \multicolumn{7}{c|}{\textbf{S-Recall} }  \\  \cline{2-15}
& Textract & vLLM-S & vLLM-C & AzureDI & Eva-D & Eva-C &  \sys & Textract & vLLM-S & vLLM-C & AzureDI & Eva-D & Eva-C & \sys \\ \hline
\code{Easy} &	0.74 & 0.68 & N/A & 0.35 & 0.27 & N/A & \textbf{0.94} & 0.66 & 0.61 & N/A & 0.38 & 0.32 & N/A & \textbf{0.92} \\ \hline
\code{Medium} & 0.61 & 0.54 & N/A & 0.42 & 0.19 & N/A & \textbf{0.81} & 0.64 & 0.49 & N/A & 0.55 & 0.2 & N/A & \textbf{0.76} \\ \hline
\code{Hard} & 0.49 & 0.51 & N/A & 0.19 & 0.18 & N/A & \textbf{0.8} & 0.36 & 0.48 & N/A & 0.17 & 0.19 & N/A & \textbf{0.76} \\ \hline
\end{tabular}
\end{adjustbox}
\caption{\small \rone{S-Precision and S-Recall on \code{Easy}, \code{Medium} and \code{Hard} Datasets on Q-Benchmark. N/A indicates that the tool is unable to produce output while preserving structures such as tables and key-value pairs.  }}
\label{table:structure-quality}
\end{table*}


\noindent\textbf{Metrics}. We report precision and recall to evaluate the quality of the extracted results. \rone{For \sys, if an extracted object is a table, we convert it to a list of key-value pairs for each cell with its corresponding key in the header. If a cell contains a missing value, we include $(key, missing)$ as part of the output. Key-value blocks naturally represent key-value pairs.} We similarly transform the extracted results from the baselines into a list of key-value pairs as the extraction result for each document per dataset. 
For a document $D$, let $KV_p$ and $KV_t$ be the inferred key-value pairs and true key-value pairs for $D$. Precision $P_D = \frac{|KV_p \cap KV_t|}{|KV_p|}$, while recall $R_D = \frac{|KV_p \cap KV_t|}{|KV_t|}$. For a dataset $\mathcal{D} = \{D_1, D_2, \dots\}$, $P_{\mathcal{D}}$ and $R_{\mathcal{D}}$ represent the average precision and recall over documents in $\mathcal{D}$. We report average precision and recall on \code{Easy}, \code{Medium}, and \code{Hard} datasets.

\rone{We further evaluate how tools predict the {\em structure} of extracted data using {\em structure-precision} ({\em S-Precision}) and {\em structure-recall} ({\em S-Recall}). We first define S-Precision and S-Recall, denoted by $SP$ and $SR$, between a true table $T$ and a predicted table $T'$, which may have different schemas, and then extend it to more complex  settings. We present an example in Appendix~\ref{subsec:appendix-structure-eval} to illustrate $SP/SR$. }

\rone{Let the precision and recall between two tuples $t_i \in T$ and $t'_j\in T'$ be $P_{ij}$ and $R_{ij}$, respectively. A tuple can be viewed as a list of key-value pairs (column-cell pairs), and let $KV_i$ and $KV_j$ be key-value pairs in $t_i$ and $t'_j$. Then,  $P_{ij} = \frac{|KV_i \cap KV_j|}{|KV_i|}$, and $R_{ij} = \frac{|KV_i \cap KV_j|}{|KV_j|}$. To move from individual tuple pairs to tables, we identify a matching $M$ of matched tuple pairs  between $T$ and $T'$.  $M$ is computed by finding a minimum-distance matching in the bipartite graph where $T$ and $T'$ each represent a partition. Each tuple is a node and the distance of edge $(t_i,t'_j)$  is $1-\frac{|KV_i \cap KV_j|}{|KV_i\cup KV_j|}$. We present a  dynamic programming algorithm to find $M$ in Appendix~\ref{subsec:appendix-structure-eval}.
Finally, $SP = \sum_{(t_i,t'_j)\in M}P_{ij}/|T'|$ and $SR = \sum_{(t_i,t'_j)\in M}P_{ij}/|T|$. For two key-value blocks, their $SP$ and $SR$ can be computed using the same metric, as a key-value block can be viewed as a table with a single row.}


\rone{We now extend $SP$ and $SR$ to handle nested and mixed data blocks. 
Intuitively, evaluating nested tables requires capturing the co-occurrence of a nested table and the tuple it is nested under. For any nested table $T$, we identify the tuple $t$ it is nested under. We then construct a flattened table from $T$ with a schema formed by concatenating the schemas of $t$ and $T$, where each tuple is the concatenation of $t$ and a tuple from $T$. For a tuple without a nested table, we preserve it by appending NULL values under the unmatched columns (i.e., a left outer-join). We recursively flatten tables when multiple levels of nesting exist. For a document $D$, we concatenate true and predicted data blocks into two wide tables $T$ and $T'$ for ease of evaluation.  Concatenating any two tables $T_i$ and $T_j$ is equivalent to performing a full outer join, i.e., combining their schemas and inserting NULLs for unmatched columns. Finally, we report $SP$ and $SR$ over the true and predicated concatenated tables $(T,T')$. }

\rone{Both $P/R$ and $SP/SR$ are useful: $SP/SR$ emphasizes structure and can heavily penalize upstream errors, while $P/R$ treats all errors equally but does not  consider structure. For example, a mistake in a single key-value pair is counted once for $P/R$ but can be penalized as many times as the number of tuples nested below it for $SP/SR$. 
}


\ifshowblock
\begin{lstlisting}
Prompt of vLLM-S: 
Please extract all key-value pairs from the following image and output only the same JSON template below. For key-value blocks, extract the pairs directly. For tables, output a key-value pair for every cell using the table headers as keys. Do not attempt to interpret the overall structure; simply extract and present all key-value pairs as they appear.
JSON Template:
[
    {
        "content": [
            {
                "type": "table",
                "content": [
                    {
                        "key1": "value1",
                        "key2": "value2"
                    },
                    {
                        "key1": "value3",
                        "key2": "value4"
                    },
                    ... more key-value pairs ...
                ]
            },
            {
                "type": "kv",
                "content": [
                    {
                        "key1": "value1"
                    },
                    ... more key-value pairs ...
                ]
            }
        ]
    }
]
\end{lstlisting}
\fi
\ifshowblock
\begin{lstlisting}
Prompt of vLLM-C: 
Please extract ALL key-value pairs from the following image and output only the same JSON template below. For key-value blocks, extract the pairs directly. For tables, output a key-value pair for every cell using the table headers as keys. Do not attempt to interpret the overall structure; simply extract and present all key-value pairs as they appear. 
JSON Template:
{
    "key1" : "value1",
    "key2" : "value2",
}
\end{lstlisting}
\fi


\vspace{-3mm}
\subsection{Experimental Results} 
\label{subsec:exp-results}
\noindent\textbf{Experiment 1: Quality Comparisons.} 
We report the precision and recall for all tools in Table~\ref{table:quality} on Q-Benchmark. In \code{Easy} datasets, both Textract and \sys perform well, with \sys outperforming Textract by around 8\% in precision and 9\% in recall. Vision-based LLM approaches, Evaporate, and AzureDI struggle even on simple templates. \update{Eva-C uses LLM-generated scripts, encoding heuristic rules, to extract field values, which turns out to be not robust. Vision-based LLMs outperform Eva-D by using images instead of text, enabling better use of visual information. 
 We will present a detailed error analysis for all baselines subsequently in Table~\ref{tab:error}. } 
\sys significantly outperforms all baselines in both  \code{Medium} and \code{Hard} datasets, achieving around \textbf{\textit{38\% higher precision and 39\% higher recall compared to Textract, 44\% higher precision and 25\% higher recall compared to AzureDI, \update{39\% higher precision and 53\% higher recall compared to the best version of Evaporate},  and 25\% higher precision and 33\% higher recall compared to the best vision-LLM approach}}.  
\rone{We also report structure-precision and structure-recall in Table~\ref{table:quality}, and \sys consistently outperforms all baselines on datasets with varying difficulty. Notably, \sys achieves 31\% higher S-Precision and 22\% higher S-Recall than the best baseline. \sys is able to infer complex nested structures (on \code{Hard} datasets), achieving 80\% S-Precision and 76\% S-Recall. 
}
All baselines struggle with complex layouts, leading to a significant drop in quality. This highlights the importance of inferring the template, making downstream extraction more accurate, compared to general-purpose data extraction solutions that ignore the  template.


\begin{figure}[tb]
    \centering
    \includegraphics[width=1\linewidth]{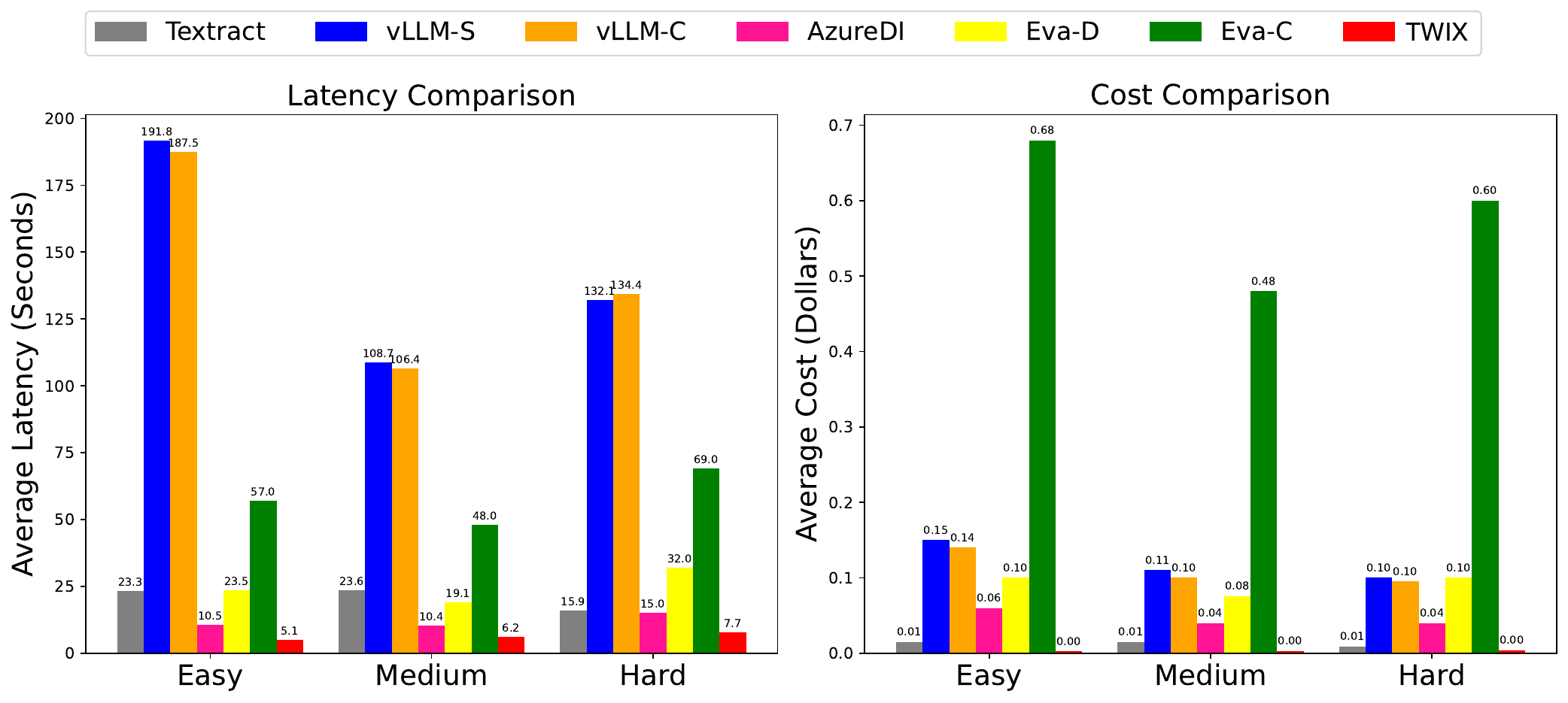}
    \vspace{-2em}
    \caption{\small Latency (Left) and Cost (Right) on Q-Benchmark.}
    \label{fig:latency_cost}
\end{figure}







\vspace{1mm}
\noindent\textbf{Experiment 2: Latency Comparisons.}
We compare the end-to-end latency \update{in Q-Benchmark} in Figure~\ref{fig:latency_cost} (left). Vision LLM approaches are time-consuming, as they process each page as an image. Their latency depends on the number of pages and the size of the output tokens. Data extraction tasks typically return a large number of output tokens, making image-based extraction costly. Textract and AzureDI are much faster than vLLM-based baselines, taking around 21 and 11s to extract data for documents with an average of 10.4 pages per dataset, respectively. \sys is the most efficient tool, taking around 5s to process a dataset—\textbf{\textit{2$\times$ faster than AzureDI,  4$\times$ faster than Textract, \update{7$\times$ faster than Evaporate}, and 28.6$\times$ faster than vLLM-based approaches}}.


\vspace{1mm}
\noindent\textbf{Experiment 3: Cost Comparisons.} 
Figure~\ref{fig:latency_cost} presents the end-to-end cost of all tools in the Q-Benchmark (right). Most baselines charge per page. Textract and AzureDI have fixed rates of \$1.50 and \$10 per 1000 pages, respectively~\cite{textract,azure}, while GPT-4-Vision APIs charge based on the number of pages and image resolution (treating each page as an image)~\cite{gptprice} and Eva-D using GPT-4o charges based on the number of tokens in the datasets. \update{Eva-C uses LLMs to generate scripts for extracting values for each field, taking raw text converted from documents as input. While data extraction using the scripts is free, the script generation step incurs high cost. }  
\sys incurs costs only during template inference, where LLMs filter non-field phrase clusters. The subsequent template-based extraction is LLM-free, incurring no further cost. Across all datasets, \sys achieves an average cost of less than \$0.003 per dataset, representing \textbf{\textit{23\% of Textract's cost, 6.4\% of AzureDI's cost, 1.2\% of Evaporate's cost}, and {\em 0.8\% of vision-LLM-based tools' cost}}.


\vspace{1mm}
\noindent\textbf{Experiment 4: Scalability Comparisons.} 
To evaluate how the tools scale to large datasets, \update{we evaluate their cost and latency on S-Benchmark, around 300$\times$ larger than Q-Benchmark, in Table~\ref{tab:scalability}. }


\sys can scale to large datasets easily with around 4 minutes and $\$$0.014 per dataset. This is because \sys infers the template from a small portion of the document determined by our pruning strategy and then extracts data for remaining documents based on the inferred template without invoking LLMs.  \sys incurs no additional cost as the number of documents grows and introduces negligible latency since template-based data extraction is inexpensive. 
\update{Eva-C uses LLMs to generate scripts that extract values for fields. Similar to \sys, their data extraction is no-cost using scripts. Even in this case, \sys still achieves noticeable advantages in both latency and cost, while also delivering over 50\% higher precision and recall based on the results in Q-Benchmark. } 
All other baselines' latency and cost increase linearly with the number of documents. Vision-based LLMs, for example, take over 30 hours to process around 2000+ pages at over \$50. \textbf{\textit{\update{\sys completes the task in a few minutes—ranging from 12$\times$ to 174$\times$ faster and 1500$\times$ to 2423$\times$ cheaper than all baselines except Eva-C.}} }



\vspace{1mm}
\noindent\textbf{Experiment 5: Time and cost breakdown of \sys.} 
\sys consists of three components: 1) phrase extraction extracting text from documents using OCR tools; 2) template inference, where \sys infers fields and structure of the template; and 3) data extraction using the inferred template. We provide a breakdown of latency for these three components in Figure~\ref{fig:breakdown} on both benchmarks. The numbers represent the percentage of latency for each component.


\begin{table}[]
    \centering
    \footnotesize
    \begin{tabular}{|c|c|c|c|c||c|c|c|c|}
    \hline
     & \multicolumn{4}{c||}{Latency (Minutes)}  & \multicolumn{4}{c|}{Cost (Dollars)} \\ \hline 
   Tools & \texttt{Easy} & \texttt{Medium} & \texttt{Hard} & \texttt{AVG} & \texttt{Easy} & \texttt{Medium} & \texttt{Hard} & \texttt{AVG}\\ \hline
     Textract & 59 & 61 & 41 & 59 & 31.8 & 32.9 & 23.7 & 31.6 \\ 
     vLLM-S & 1977 & 1810 & 876  & 1901 & 54.5 & 52.3 & 43 & 54 \\ 
     vLLM-C & 2195 & 2119 & 1008  & 2133 & 53.2  & 50.8 & 41.9 & 52.4 \\ 
     AzureDI & 82 & 49 & 30 & 74 & 21.1 & 21.9 & 15.8 & 21.1 \\
     Eva-D & 48 & 54 & 47 & 48 & 24.5 & 27.1 & 13.6 & 24.1 \\ 
     Eva-C & 9 & 11 & 10 & 9 & 4.3 & 4.2 & 3.9 & 4.3 \\ 
     \sys & \textbf{4.2} & \textbf{4.3}  & \textbf{3.3} & \textbf{4.1}  & \textbf{0.014} & \textbf{0.015}  & \textbf{0.012} & \textbf{0.014} \\ \hline 
    \end{tabular}
 \caption{\small Latency and Cost on S-Benchmark.} 
    \vspace{-2em}
    \label{tab:scalability}
\end{table}

\begin{figure*}[t]
  \centering
  \begin{minipage}[t]{0.49\textwidth}
    \centering
    \includegraphics[width=6.6cm,height=5cm]{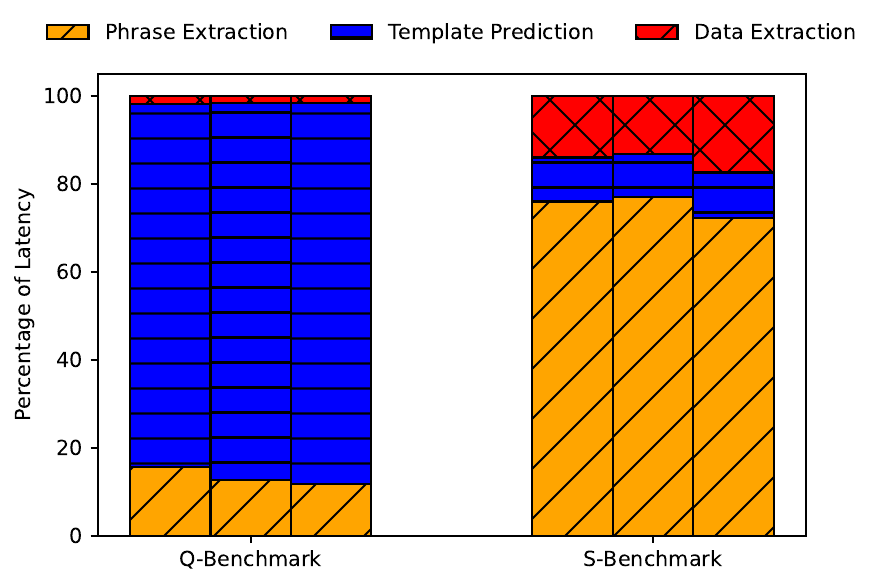}
    \vspace{-1em}
    \caption{\small Latency Breakdown.}
    \label{fig:breakdown}
  \end{minipage}%
  \hfill
  \begin{minipage}[t]{0.48\textwidth}
    \centering
    \includegraphics[width=6.5cm,height=5cm]{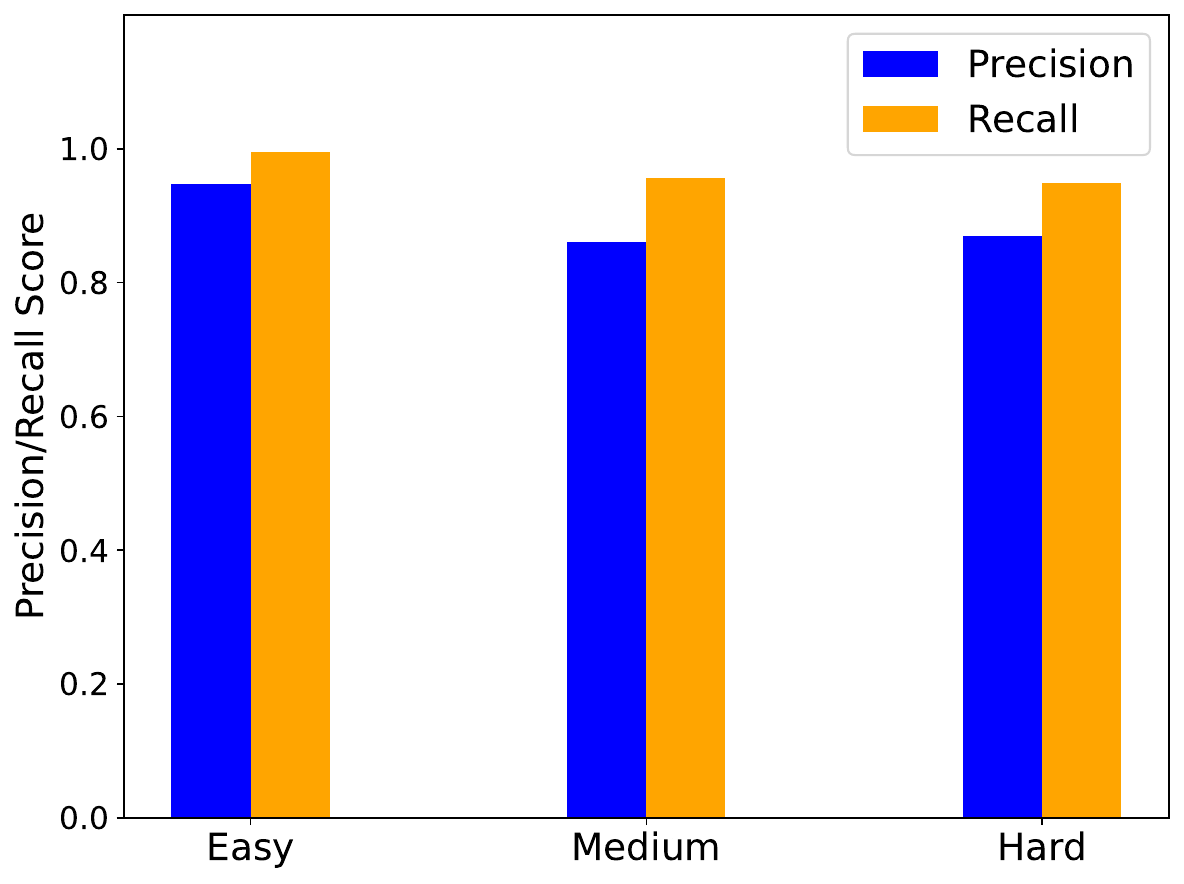}
    \vspace{-1em}
    \caption{\small Field Inference.}
    \label{fig:field-prediction}
  \end{minipage}
\end{figure*}


When the dataset is relatively small in Q-Benchmark,  template inference consumes around 85\% of the time in \code{Easy}, \code{Medium}, and \code{Hard} datasets, and data extraction is the fastest component. However, as the dataset size grows, \update{as in S-Benchmark}, the most time-consuming component becomes phrase extraction, as its time increases linearly with the number of documents, while the latency of template inference remains constant since \sys infers the template from a small document portion, thanks to the effective pruning strategy.  \update{In terms of cost breakdown, all costs in \sys are incurred during the template inference stage where LLMs are used. }


\vspace{1mm}
\noindent\textbf{Experiment 6: Performance of Field Inference}. We examine the performance of the field inference in \sys for Q-Benchmark, and report precision and recall in Figure~\ref{fig:field-prediction}. 

In \code{Easy} datasets with simple template structure, the precision and recall of the inferred fields are around 0.95 and 0.98, respectively. When the template becomes more complex, as in the \code{Medium} and \code{Hard} datasets, the precision drops to around 0.86 while the recall still remains close to 0.95. A high recall is important for predicting the template structure since the set of key rows will likely be recovered, which are the backbone of tables.

Template structure inference helps correct field inference errors. First, metadata (headers, footers, etc.) incorrectly inferred as keys or values are often corrected, as they rarely align with rows in table or key-value blocks, violating constraints in the row labeling problem. Second, false positives of fields are frequently corrected after row labeling. For example, the false positives in a value row will be corrected once the row is identified as a value row. These observations highlight \sys's robustness in field inference.




\noindent\textbf{Experiment 7: Error Pattern Analysis}. We identify six common error types observed in the compared tools below and summarize them  in Table~\ref{tab:error}, where we use $\checkmark\checkmark$ and $\checkmark$ to indicate errors that {\em most significantly affect} and {\em moderately affect} each tool, respectively. 
\techreport{Note that all tools exhibit some degree of error in every type. For errors that have minimal impact on a tool's performance, no corresponding check mark is included in the table}.


\noindent\textbf{Type (1): Column Misalignment}. When multiple consecutive values are missing in a row in a table block, vision LLMs \update{and Evaporate} often fail to count the exact number of missing values, resulting in mismatched values and rows.
\textbf{Type (2): Misidentification of Table Headers}. AzureDI fails to identify the table header correctly, instead selecting the first row as the header in 4 out of 34 datasets. Textract similarly misidentifies table headers, particularly in documents with complex layouts. 
\textbf{Type (3): Misidentification of Data Blocks}. All baselines struggle to accurately identify data blocks in complex layouts, such as in \code{Medium} and \code{Hard} datasets. They may miss data blocks entirely or merge multiple blocks into one. 
For example, vision-based LLMs show inconsistent behavior, occasionally extracting partial data or missing entire blocks. 
\textbf{Type (4): Phrase Extraction Errors from OCR}. OCR-based phrase extraction is imperfect and affects all tools. OCR may merge closely spaced phrases into one or split a long phrase into multiple parts. While such errors are relatively infrequent, they can still impact results, particularly when fields like table headers are extracted incorrectly. 
\textbf{Type (5): Metadata Misclassified as Fields or Values}. If a metadata row is misclassified as a value row and aligns with a table header, \sys incorrectly includes it in the table block, causing false positives. Similar behavior is observed in other tools; vision-LLMs frequently create new fields for metadata and incorporate them into data blocks. 
\textbf{Type (6): Field Inference Errors}. None of the tools achieves perfect field inference, which affects downstream tasks such as table or key-value block inference in the baselines or row labeling in \sys.

\begin{table}[]
    \centering
    \footnotesize
    \begin{tabular}{|c|c|c|c|c|c|c|}
    \hline
     & \textbf{Type 1} & \textbf{Type 2} & \textbf{Type 3} & \textbf{Type 4} & \textbf{Type 5} & \textbf{Type 6} \\ \hline
     Textract &  & $\checkmark\checkmark$ & $\checkmark\checkmark$ & $\checkmark$& $\checkmark$ & $\checkmark\checkmark$ \\ \hline
     vLLM-S & $\checkmark\checkmark$ & & $\checkmark\checkmark$ & $\checkmark$& $\checkmark$ & $\checkmark$\\ \hline
     vLLM-C & $\checkmark\checkmark$ & & $\checkmark\checkmark$ & $\checkmark$& $\checkmark$ & $\checkmark$\\ \hline
     AzureDI & & $\checkmark\checkmark$ & $\checkmark\checkmark$ & $\checkmark$ & $\checkmark$ & $\checkmark\checkmark$ \\ \hline
     Eva-D & $\checkmark$ &  & $\checkmark$ & $\checkmark$& $\checkmark\checkmark$ & $\checkmark\checkmark$ \\ \hline 
     Eva-C & $\checkmark\checkmark$ &  & $\checkmark\checkmark$ &$\checkmark$ & $\checkmark$ & $\checkmark\checkmark$ \\ \hline 
     \sys & & &  & $\checkmark$ & $\checkmark$ & $\checkmark\checkmark$ \\ \hline
    \end{tabular}
 \caption{\small Frequent Error Patterns of Compared Tools ($\checkmark\checkmark$ indicates frequent, and $\checkmark$ indicates sometimes). } 
    \vspace{-0.4em}
    \label{tab:error}
\end{table}

\rtwo{
\topic{OCR Error Analysis} We describe below the three most common OCR errors, their impact on \sys, and strategies to handle them. 
{\bf 1) Phrase Misspelling.} OCR may produce misspelled phrases (e.g., \code{Date} misread as \code{Dafe}). If the error occurs in a value, the effect is local and minimal. If it occurs in a field, the errors are  {\em consistent} across records due to the shared template (e.g., all instances of "Date" are misread as "Dafe"). This consistency allows \sys to correctly infer the template. LLMs can be used during post-processing to semantically correct such errors in fields, though that is beyond the scope of \sys. 
{\bf 2) Bounding Box Extraction Inaccuracy.} \sys uses bounding boxes for row alignments when inferring template structure. Bounding boxes, if
imprecise, are typically only slightly over-sized and are still reliable to determine vertical alignment.  {\bf 3) Multi-row Phrase Extraction. } When a phrase spans multiple rows within a table cell, OCR may extract it as separate phrases. While improving OCR is beyond the scope of \sys, multi-row phrases can be extracted by either using visual cues (e.g., cell borders or lines) to merge the fragments into a single phrase or applying semantic checks with LLMs to determine whether the split phrases should be merged.  }

%% file: relatedwork.tex
\section{Related Work}
\label{sec:relatedwork}

Our work is relevant to document and web data extraction, as well as document layout analysis. 

\noindent\textbf{Document Data Extraction}. 
There have been multiple papers~\cite{arora2023language,tata2021glean,aggarwal2021form2seq,sarkhel2021improving} and industrial tools~\cite{textract,google-doc,microsoft-doc} aimed at extracting data from documents. Most extraction tools,  e.g.,~\cite{textract,google-doc,microsoft-doc,aggarwal2021form2seq,sarkhel2021improving,arora2023language} are general-purpose solutions for extracting tables or entities from visually rich form-like documents. \update{Evaporate~\cite{arora2023language} uses LLMs to generate scripts that are then used for data extraction. These scripts encode heuristic rules that are not robust for reliable extraction. }

Tools from industry like Textract~\cite{textract}, Google Document AI~\cite{google-doc}, and Azure Document Intelligence~\cite{microsoft-doc} use pretrained models to extract structured data, such as tables or key-value pairs, from form-like documents. These tools perform well on simple layouts and domains well-represented in their training data (e.g., Google Document AI offers models for specific domains like taxes or invoices). However, they often fail on unseen documents with complex layouts and incur significant costs and latency. Recent advances in vision-based LLMs, such as GPT-4 Vision~\cite{gpt4vision}, show promise but lack consistent performance, with high latency and substantial costs.


Learning-based extraction~\cite{tata2021glean,or2021few,katti2018chargrid,bai2017neural,majumder2017deep,paliwal2019tablenet,le2014flashextract,yang2022survey} retrieves values for user-specified fields from documents by training deep learning models on human-labeled data. However, it requires significant human effort (e.g., specifying fields and labels) and doesn't capture relationships between extracted field values, as discussed in Section~\ref{sec:introduction}.  In contrast, our unsupervised approach uncovers the template—the backbone of templatized documents—and efficiently extracts structured data. For example, Parthasarathy et al.~\cite{parthasarathy2022landmarks} introduce the concept of landmarks to narrow down the region of interest in a document and then extract values for user-specified fields, while \sys preserves relationships in the extracted data.



\noindent\textbf{Web Data Extraction}. Our work is also related to web data extraction, which primarily relies on HTML tags~\cite{parameswaran2011optimal,dalvi2009robust,chen2022web,cetorelli2021smallest,dalvi2011automatic,dalvi2009web,bohannon2012automatic,dalvi2012analysis,etzioni2004web,agichtein2000snowball,sarawagi2008information,sarkhel2021improving,niu2012deepdive,arasu2003extracting,crescenzi2004automatic,kayed2009fivatech,liu2009vide}. 
The earliest work in this vein analyzed differences
between pages generated using the same HTML template to learn the template through techniques like regular expressions~\cite{arasu2003extracting,crescenzi2004automatic,kayed2009fivatech,liu2009vide}.
Other papers extended
this to  develop robust wrappers to handle the evolution of underlying HTML templates in web pages over time~\cite{parameswaran2011optimal,dalvi2009robust}, while others explored generating domain-centric wrappers for web-scale information extraction, designed to tolerate noise~\cite{dalvi2011automatic,dalvi2009web,bohannon2012automatic,dalvi2012analysis}. 
For example, Miria~\cite{chen2022web} extracts records from websites by identifying invariants across records based on HTML tag trees. Cetorelli et al.~\cite{cetorelli2021smallest} introduce a landmark-based grammar from a set of web pages with a common HTML template.  To extract relations from semi-structured web data, Lockard et al.~\cite{sarkhel2021improving} proposed a distant supervision approach, while DeepDive~\cite{niu2012deepdive} further leveraged XML and HTML-specific feature descriptors. 
However, HTML or XML tags present in web-pages or semistructured documents indicating nesting relationships are often not available in documents like PDFs. 

Several studies extract tables from the web without relying on HTML tags. Chu et al.~\cite{chu2015tegra} split phrases in record rows into cells aligned with corresponding columns, while Chen et al.~\cite{chen2013automatic} extract information from spreadsheets by leveraging font formatting to extract metadata and using learning-based methods to label rows. Cafarella et al.~\cite{cafarella2009web} extract fact tuples by parsing natural language sentences on the web. Finally, Gao et al.~\cite{gao2018navigating} explore an unsupervised approach to extract structures for log datasets.  These methods primarily focus on data with simple layouts or structures, like tables or logs, whereas our approach handles complex visual layouts, including nested combinations of tables and key-value structures.



\noindent\textbf{Document Layout Analysis}. Document layout analysis (DLA)~\cite{soto2019visual,xu2020layoutlm,binmakhashen2019document,long2022towards,patil2020read,tang2023unifying,breuel2003high,yang2017learning} detects various document layouts, such as pages, texts, tables, images, titles, headers, and footers, using visual features (e.g., font size and type), content, and structural patterns. While effective for coarse-grained components like text and table blocks, DLA struggles with fine-grained components, such as mixed key-value and table blocks. Combining DLA with our approach could potentially enhance data extraction by first detecting structured portions in long documents with text or images, allowing our method to handle the structured parts and expanding its applicability. This remains an interesting avenue for future work.









%% file: conclusion.tex
\section{Conclusion}
\label{sec:conclusion}
We present \sys, a robust, efficient, and effective tool to extract structured data from a collection of templatized documents. 
\sys first infers a flexible visual template used to create documents, using which
it separates templatized documents into nested records and data blocks within records, enabling accurate and efficient data extraction from each block. We demonstrate hardness for the underlying problems, while also providing correctness guarantees. \sys combines an optimization approach for principled template discovery, while leveraging LLMs to  provide semantic knowledge in carefully targeted ways. 
 Our experiments show that \sys outperforms baselines significantly across accuracy, latency, and cost.




%% file: ack.tex
\subsection*{Acknowledgments}

We acknowledge support from grants DGE-2243822, IIS-2129008, IIS-1940759, IIS-1955488, IIS-2027575, and IIS-1940757 awarded by the National Science Foundation, DOE award DE-SC0016260, AC02-05CH11231, DARPA Agreement No. HR00112590131, funds from the State of California, an NDSEG Fellowship, funds from the Alfred P. Sloan Foundation, as well as EPIC lab sponsors: Adobe, Bridgewater, Google, G-Research, Microsoft, PromptQL, Sigma Computing, and Snowflake. Compute credits were provided by Azure, Modal, NSF (via NAIRR), and OpenAI. 

%% file: appendix.tex
\clearpage
\appendix
\section{Appendix}

\subsection{Proof of Proposition 1}
\label{subsec:proof-proposition1}

\setcounter{proposition}{0}
\begin{proposition}
    Given the true template $T' = (V', E')$, when there exists a unique node  $v\in V', v.type = table$, $p_i,p_j \in v.fields$, if $L_{p_i} = L_{p_j}$, then $p_i$ and $p_j$ are a perfect match; if $L_{p_i} > L_{p_j}$, then $p_j$ and $p_j$ are a partial perfect  match. 
\end{proposition}

\noindent\begin{proof}[Proof of Proposition~\ref{prop:table}]
In the true template $T' = (V', E')$, when there exists a unique node  $v\in V', v.type = table$, $p_i,p_j \in v.fields$, if $L_{p_i} = L_{p_j}$, we have $\nexists v'\in V'$, $v \neq v'$ , s.t., $p_i \in v'.fields$ or $p_j\in v'.fields$, under the assumption that we make that each phrase $p$ has a unique label. Otherwise, $L_{p_i} \neq L_{p_j}$. Let the location vectors of $p_i$ and $p_j$ be $v_{p_i} = [i_1,i_2,...,i_m]$, $v_{p_j} = [j_1,j_2,...,j_m]$, respectively.   $\forall B_k$, where $v\rightarrow B_k$, let the index of $p_i$ and $p_j$ in $B_k$ be $i_k$ and $j_k$, respectively. $\forall k_1,k_2 \in [1,m]$, we have $i_{k_1} - j_{k_1} = i_{k_2} - j_{k_2}$, since $v\rightarrow B_{k_1}$ and $v\rightarrow B_{k_2}$, and the schema of table in the template is consistent across the records. This completes the first half of proposition. 

When  $L_{p_i} > L_{p_j}$, let $v_{p_i}^{'}$ be the subsequence of $v_{p_i}$ that occurs in blocks created from $v$, i.e., $\forall p_k \in v_{p_i}^{'}, p_k \in B$, where $v\rightarrow B$. Based on the above proof, $v_{p_i}^{'}$ is a perfect match with $v_{p_j}$, and thus $v_{p_i}$ is a partial perfect match with $v_{p_j}$. 
\end{proof}

\subsection{Proof of Proposition 2}
\label{subsec:proof-proposition2}

\setcounter{proposition}{1}
\begin{proposition}
\label{prop:kv}
    Consider the true template $T' = (V',E')$. Given a unique node  $v\in V', v.type = Key$-$Value$, $p_i,p_j \in v.fields$, and $f(p_i)=f(p_j)=True$, if $L_{p_i} = L_{p_j}$, then $p_i$ and $p_j$ are a perfect match; if $L_{p_i} > L_{p_j}$, then $p_j$ and $p_j$ are a partial perfect match. 
\end{proposition}

\noindent\begin{proof}[Proof of Proposition~\ref{prop:kv}]
Consider the location vectors of $p_i$ and $p_j$, $v_{p_i} = [i_1,i_2,...,i_m]$, $v_{p_j} = [j_1,j_2,...,j_m]$. 
    Under the assumption that a phrase $p$ has a unique label, when there exists a unique node  $v\in V', v.type = Key$-$Value$, $p_i,p_j \in v.fields$, and $f(p_i)=f(p_j)=True$, if $L_{p_i} = L_{p_j}$, $\forall B_k$, $v\rightarrow B_k$, let the index of $p_i$ and $p_j$ in $B_k$ be $i_k$ and $j_k$, respectively. 
    
    $\forall k_1,k_2 \in [1,m]$, $f(p_i)=f(p_j)=True$ implies that the values of $p_i$ and $p_j$ are consistently filled or missing across the records from the same key-value node in the template. Additionally, since the list of fields in two key-value blocks generated from the same key-value node are consistent across records, we have   $i_{k_1} - j_{k_1} = i_{k_2} - j_{k_2}$. Thus $p_i$ is a perfect match of $p_j$. When $L_{p_i} > L_{p_j}$, let $v_{p_i}^{'}$ be the subsequence of $v_{p_i}$ that occurs in blocks created from $v$, i.e., $\forall p_k \in v_{p_i}^{'}, p_k \in B$, where $v\rightarrow B$. Based on the above proof, $v_{p_i}^{'}$ is a perfect match with $v_{p_j}$, and thus $v_{p_i}$ is a partial perfect match with $v_{p_j}$. 
\end{proof}

\subsection{Proof of Theorem 1}
\label{subsec:proof-theorem1}

\setcounter{theorem}{0}
\begin{theorem}
\label{theo:problem}
    The following is equivalent to Eq.~(\ref{equ:goal}), \\ 
    \begin{align}\label{equ:obj}
        \max \sum_{r_i \in R} (& y_i^{K}\text{log}(Prob_{K}^{r_i}) + y_i^{V}\text{log} + y_i^{KV}\text{log}(Prob_{KV}^{r_i}) +  \nonumber \\    &y_i^{M}\text{log}(Prob_{M}^{r_i}))   \tag{9}
    \end{align}
\end{theorem}

\begin{proof}
First, we take $\log$ on objective~(\ref{equ:goal}), resulting in: 
\begin{equation}
\label{equ:log-objective}
    \max \sum_{r_i \in R} \log ( y_i^K \text{Prob}_K^{r_i} + y_i^V \text{Prob}_V^{r_i} + y_i^{KV} \text{Prob}_{KV}^{r_i} \\+ y_i^M \text{Prob}_M^{r_i})
\end{equation}
The objective in (\ref{equ:log-objective}) is equivalent to~(\ref{equ:obj}) since logarithms are monotonically increasing. However, this objective is still non-linear. Given the constraint $\forall r_i \in R, y_i^{K} + y_i^{V} + y_i^{KV} + y_i^{M} = 1$, a row $r_i$ can take exactly one label. For simplicity, let $z_i = \log(y_i^K \text{Prob}_K^{r_i} + y_i^V \text{Prob}_V^{r_i} + y_i^{KV} \text{Prob}_{KV}^{r_i} + y_i^M \text{Prob}_M^{r_i})$. If $y_i^{K} = 1$, then all other $y_i$, such as $y_i^{V}$, are 0. In this case, $z_i = y_i^{K}\log(\text{Prob}_K^{r_i})$. Similarly, when $y_i^{V} = 1$, $z_i = y_i^{V}\log(\text{Prob}_V^{r_i})$. This makes the objective in~(\ref{equ:log-objective}), $\max \sum_{r_i \in R} z_i$, equivalent to the linear objective function in~(\ref{equ:goal}). 
\end{proof}

\subsection{Proof of Theorem 2}
\label{subsec:proof-np-hard}

\setcounter{theorem}{1}
\begin{theorem}
\label{theo:np-hard}
    {\sc Row labeling} is {\sc NP-hard}. 
\end{theorem}

\input{npproof}
\subsection{Proof of Theorem 3}
\label{subsec:proof-theorem3}

\setcounter{theorem}{2}
\begin{theorem}
\label{theo:input}
    Let $F'$ be the true fields in the true template $T' = (V',E')$. If $\exists p\in F \cap F'$, and $\forall Rec$, $p$ appears exactly once in $Rec$, then $R'$ contains at least one record $Rec$ created from $T'$. 
\end{theorem}

\begin{proof}[Proof of Theorem~\ref{theo:input}]~\label{proof:input}
If there exist a correctly inferred field $p\in F$ that appears exactly once in every record, and $R'$ contains every field at least twice, then $R'$ contains at least one complete record given that $R'$ is a consecutive sublist of rows. 
\end{proof}

\subsection{Representation of Data Extraction Results} 
\label{susbec:data-extraction-objects}

\begin{figure}[tb]
    \centering
    \includegraphics[width=0.8\linewidth]{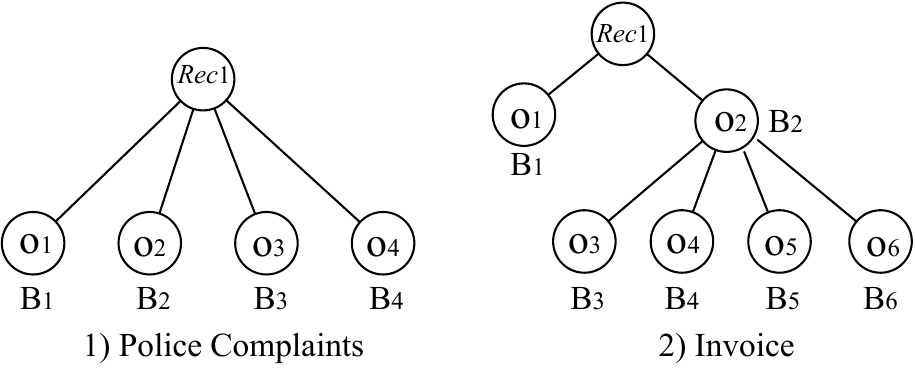}
    \caption{\small Data Extraction Objects of the First Record for Police Complaints and Invoices.}
    \label{fig:extraction-objects}
\end{figure}

We represent the extracted data per record into a tree.  Let $o_i$ denote the {\em data extraction object} of a data block $B_i$, associated with three attributes: $o_i.fields$, $o_i.type$, and $o_i.content$. Here, $o_i.fields = B_i.fields$, and $o_i.type \in \{\textit{Table}, \textit{Key-Value}\}$ based on the type of $B_i$. When $o_i.type = \textit{Key-Value}$, $o_i.content$ is a list of key-value pairs. When $o_i.type = \textit{Table}$, $o_i.content$ represents the extracted table, with $o_i.fields$ as the header and the corresponding extracted tuples. 

 \begin{figure*}[tb]
    \centering
    \includegraphics[width=0.7\linewidth]{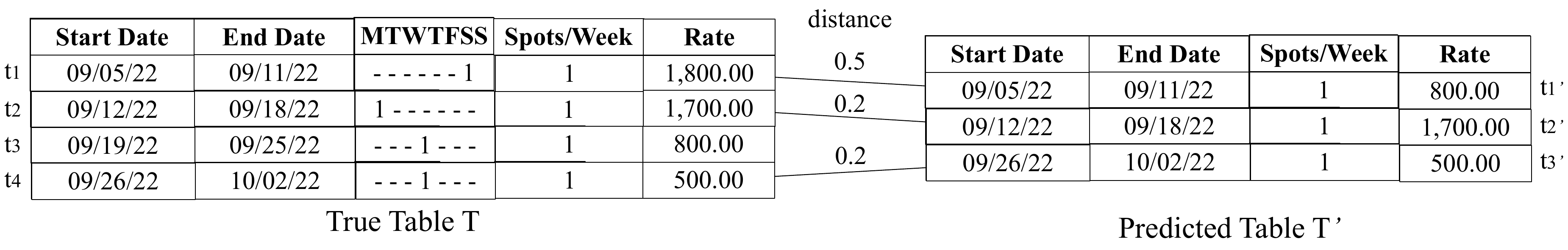}
    \vspace{-1em}
    \caption{\small \rone{Evaluating tables with variable schemas.} }
    \vspace{-1em}
    \label{fig:table-eval} 
\end{figure*}

\begin{figure*}[tb]
    \centering
    \includegraphics[width=1\linewidth]{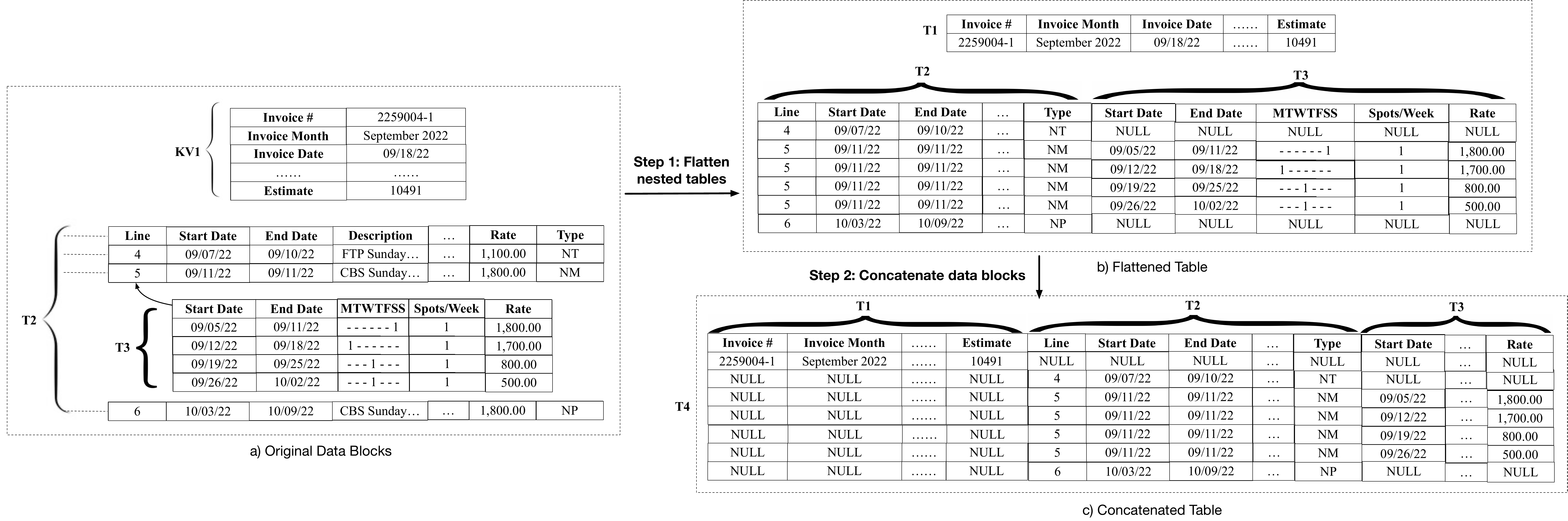}
    \caption{\small \rone{Data Transformation for Structure Evaluation.} }
    \label{fig:unnest} 
\end{figure*}

Let $O$ be the {\em data extraction object} of a record $Rec$, and $O$ is a tree, where the nodes are the set of data extraction objects of blocks in $Rec$. For two nodes $o_i,o_j$, if $o_i$ is a parent of $o_j$, then $B_i \cap B_j \neq \emptyset$ and $loc(B_i) < loc(B_i)$, i.e., $B_i$ appears before $B_j$ and $B_i$ overlaps with $B_j$. If $o_i$ is a left sibling of $o_j$, then $B_i \cap B_j = \emptyset$ and $loc(B_i) < loc(B_i)$.

\begin{example}
    The data extraction objects of the first record for police complaints and invoices are presented in Figure~\ref{fig:extraction-objects}. For $Rec_1$ in police complaints, four data extraction objects $o_1,o_2,o_3$ and $o_4$ are sequentially placed under the root, while in invoices, $o_2$ is the parent of $o_3$ as their blocks $B_2$ and $B_3$ overlaps and $B_2$ appears before $B_3$ in Figure~\ref{fig:invoice}. 
\end{example}

\subsection{Correctness Analysis of \sys in Compliant Documents}
\label{subsec:appendix-correctness}

We analyze the correctness of the returned data extraction objects in one commonly observed type of document that we call  {\em compliant documents}, defined below. 

\setcounter{definition}{0}
\begin{definition}
    [Compliant Document]. Consider a table block $B$. If every value phrase is uniquely vertically aligned with its field phrase, then $B$ is a \textit{compliant table block}. Consider a key-value block $B = [p_1,p_2, ..., p_m]$, where $p_i$ is a phrase.  If $\forall p_i\in B$, $p_i\notin B.fields$, we have $p_{i-1} \in B.fields$, then $B$ is a {\em compliant key-value block}. Given a document $D$, if $\forall B \in D$, $B$ is compliant, then $D$ is a {\em compliant document}. 
\end{definition}

In a compliant table block, every value row is vertically aligned with the header row in the table, while in a compliant key-value block $B$,  every value $p\notin B.fields$ has a preceding field.

\begin{theorem}
\label{theo:extraction}
For a compliant document $D$ whose Metadata rows are not well-aligned with any Key rows, the data extraction objects for $D$ are correct as long as the phrases $P$ in $D$ are extracted correctly and the inference of $T$ is correct. 
\end{theorem}

Many real-world templatized documents in our benchmark are complex due to intricate template structures. However, most individual data blocks—whether table or key-value blocks—are compliant. Across 34 real-world datasets with over 1,000 documents spanning diverse domains, 91\% of the documents are compliant. For the non-compliant cases, removing just 3\% of non-compliant phrases results in compliant documents. A non-compliant phrase refers to a table cell misaligned with its field or a value in a key-value block without a corresponding field. 
Thus, Theorem~\ref{theo:extraction} provides correctness guarantees for our approach across extensive real-world datasets, with a detailed proof below. At the end of Section~\ref{sec:exp}, we present an empirical analysis of scenarios where \sys fails. 
\begin{proof}[Proof of Theorem~\ref{theo:extraction}]~\label{proof:extraction}
If the inference of $T$ is correct, then the set of fields is accurately inferred for each data block. Any metadata row $r$ not aligned with a key row will be correctly identified as metadata in Algorithm~\ref{alg:block}. Furthermore, data block separation is correct if the inferred template $T$ matches the true template, as all records are generated using the same $T$. 
In a compliant table block, \sys correctly extracts all field-value mappings in the key row and its value rows since each value phrase is uniquely vertically aligned with its corresponding field phrase. Similarly, \sys accurately extracts key-value pairs in a compliant key-value block because every true value has a preceding key. This concludes the proof. 
\end{proof}

\subsection{Structure Evaluation}
\label{subsec:appendix-structure-eval}

\subsubsection{Example of Structure Evaluation}
\rone{We present the example to illustrate how to compute structure precision and recall, $SP$ and $SR$, defined in Section~\ref{subsec:exp-setup}, between a true table  $T$ and a predicted table $T'$, which may have different schemas. }

\topic{\rone{S-Precision and S-Recall}}  \rone{Let the precision and recall between two tuples $t_i \in T$ and $t'_j\in T'$ be $P_{ij}$ and $R_{ij}$, respectively. A tuple can be viewed as a list of key-value pairs (column-cell pairs), and let $KV_i$ and $KV_j$ be key-value pairs in $t_i$ and $t'_j$. Then,  $P_{ij} = \frac{|KV_i \cap KV_j|}{|KV_i|}$, and $R_{ij} = \frac{|KV_i \cap KV_j|}{|KV_j|}$. In Figure~\ref{fig:table-eval}, $P_{11} = 0.75$ and $R_{11} = 0.6$. To move from individual tuple pairs to tables, we identify a list of matched tuple pairs $M$ between $T$ and $T'$.  $M$ is computed by finding a minimum-distance matching in the bipartite graph where $T$ and $T'$ each represent a partition. Each tuple is a node and the distance of edge $(t_i,t'_j)$  is $1-\frac{|KV_i \cap KV_j|}{|KV_i\cup KV_j|}$. $M$ also preserves tuple order: a match as shown in Figure~\ref{fig:table-eval} is valid, whereas a match like $(t_1, t'_3), (t_2, t'_1)$ is not. 
We present a  dynamic programming algorithm to find $M$ in Algorithm~\ref{alg:tablematch} illustrated shortly.   Finally, $SP = \frac{\sum_{(t_i,t'_j)\in M}P_{ij}}{|T'|}$ and $SR = \frac{\sum_{(t_i,t'_j)\in M}P_{ij}}{|T|}$. For two key-value blocks, their $SP$ and $SR$ can be computed using the same metric, as a key-value block can be viewed as a table with a single row. For example, in Figure~\ref{fig:unnest}, a key-value block $KV_1$ is treated as a special table $T_1$. We now extend $SP$ and $SR$ to handle nested and mixed data blocks. }

\topic{\rone{Flattening Nested Tables}} \rone{Intuitively, evaluating nested tables requires capturing the co-occurrence of a nested table and the tuple it is nested under. For any nested table $T$, we identify the tuple it is nested under (e.g., $T_3$ in Figure~\ref{fig:unnest} is nested under the second tuple in $T_2$). We then construct a flattened table (in Figure~\ref{fig:unnest}-b) from $T$ with a schema formed by concatenating the schema of $t$ (from $T_2$) and $T$ (from $T_3$), where each tuple is the concatenation of $t$ and a tuple from $T$. In this way, $t$ co-occurs with every tuple from the nested table $T$ under it. For a tuple without a nested table (e.g., the first and third tuples in $T_2$), we perform a left  outer join to preserve this tuple by appending NULL values under the unmatched columns. We recursively flatten tables when multiple levels of nesting exist.}

\topic{\rone{Concatenating Data Blocks}}
\rone{For a document $D$, we convert the true and predicted data blocks into two wide tables $T$ and $T'$ for ease of evaluation. Concatenating any two tables $T_i$ and $T_j$ is equivalent to performing a full outer join, i.e., combining their schemas and inserting NULLs for unmatched columns. For example, $T_4$ in Figure~\ref{fig:unnest}-c is the concatenated table, a full outer-join of $T_1$ and the flattened table from $T_2$ and $T_3$ in Figure~\ref{fig:unnest}-b. 
Then, when evaluating precision and recall (and distance) for any tuple pair $(t_i, t'_j)$, where $t_i \in T$ and $t'_j \in T'$, we only retain the non-NULL values. Searching for a minimum-distance match between $T$ and $T'$ finds the {\em most similar} predicted blocks for each true block and returns their $SP$ and $SR$.    }

\rone{Finally, we report $SP$ and $SR$ over the concatenated tables from true and predicted data blocks, $T$ and $T'$, as the structure precision and recall for document $D$. }

\vspace{-2em}
\subsubsection{Table Match Algorithm}
\revised{We present a dynamic programming algorithm in Algorithm~\ref{alg:tablematch} to find the matched tuples between two tables $T$ and $T'$. A tuple can be viewed as a list of key-value pairs, where the key is a column and the value is its corresponding table cell. We compute the Jaccard similarity between two tuples $t_i \in T$ and $t_j \in T'$ as $Jaccard(t_i, t_j) = \frac{KV_i \cap KV_j}{KV_i \cup KV_j}$, where $KV_i$ and $KV_j$ are the sets of key-value pairs in $t_i$ and $t_j$, respectively. We call $1-Jaccard(t_i, t_j)$ as the distance between $t_i$ and $t_j$. Intuitively, matched tuple pairs in $T$ and $T'$ should have high similarity, and the list of matched pairs must preserve the tuple order in both tables. For example, a match as shown in Figure~\ref{fig:table-eval} is valid, whereas a match like $(t_1, t'_3), (t_2, t'_1)$ is not.  }

\revised{Let $dp[i][j]$ be the minimum total distance to align the first $i$ tuples in $T$ to a subsequence of tuples ending at the $j$-th tuple in $T'$. Let $t_i\in T$ and $t'_j \in T'$ be the $i$-th and $j$-th tuple in $T$ and $T'$, respectively. 
The algorithm initializes all $dp[i][j]$ values to infinity and sets $dp[1][j] = 1 - \text{Jaccard}(t_1, t'_j)$ by definition.  We additionally maintain $parent[i][j]$ to store the tuple pair that yields the minimum total distance, enabling recovery of the matched tuples via backtracking (Line 2-10).  To compute $dp[i][j]$, we consider the minimum total cost of $dp[i][k]$ over all $k < j$, since the alignment must preserve order. Thus, $dp[i][j] = \min_{k < j} \left\{ dp[i-1][k] + 1 - \text{Jaccard}(t_i, t'_j) \right\}$ (Line 11-18). Finally, we recover the matched tuples by backtracking through the stored $parent$ values (Line 19-25) }

\setlength{\textfloatsep}{0pt}
\begin{algorithm}[bt]
    \small
    \caption{\textsc{TableMatch}$(T,T')$}
    \label{alg:tablematch}
\textbf{Input:} $T,T'$\\
$n \leftarrow |T|, m\leftarrow|T'|;$\\ 

\textit{// initialization}\\
\For{$i\gets 0$ \textbf{to} $n$}{
    \For{$j\gets 0$ \textbf{to} $m$}{
        $dp[i][j]\leftarrow +\infty$\;$parent[i][j]\leftarrow -1$
    }
}

\For{$j\gets 1$ \textbf{to} $m$}{
    $dp[1][j]\leftarrow (1-Jaccard(t_1,t'_j))$\;$parent[1][j]\leftarrow 0$
}

\textit{// transitions}\\
\For{$i\gets 2$ \textbf{to} $n$}{
    \For{$j\gets 1$ \textbf{to} $m$}{
        \For{$k\gets 1$ \textbf{to} $j-1$}{
            $cost\leftarrow dp[i-1][k]+(1-Jaccard(t_i,t'_j))$\\
            \If{$cost<dp[i][j]$}{
                $dp[i][j]\leftarrow cost$\;$parent[i][j]\leftarrow k$
            }
        }
    }
}

\textit{// recover match}\\
$match\leftarrow[\,]$\;$j\leftarrow best\_col$\\
\For{$i\gets m$ \textbf{downto} $1$}{
    $match \leftarrow match \cup (i,j)$\;$j\leftarrow parent[i][j]$
}

\textbf{return} $match$

\end{algorithm}

\subsection{Prompts of vision-based LLMs}
\label{subsec:appendix-prompts}
\rone{The prompts for \code{vLLM-C} and \code{vLLM-S} take an image as input, while those for \code{LLM-C} and \code{LLM-S} take text as input.}

\begin{figure}[H]
\vspace{1mm}
\centering
\begin{lstlisting}
Prompt of vLLM-C: 
Please extract ALL key-value pairs from the following image and output only the same JSON template below. For key-value blocks, extract the pairs directly. For tables, output a key-value pair for every cell using the table headers as keys. Do not attempt to interpret the overall structure; simply extract and present all key-value pairs as they appear. 
JSON Template:
{
    "key1" : "value1",
    "key2" : "value2",
}
\end{lstlisting}
\end{figure}

\begin{figure}[H]
\centering
\begin{lstlisting}
Prompt of vLLM-S: 
Please extract all key-value pairs from the following image and output only the same JSON template below. For key-value blocks, extract the pairs directly. For tables, output a key-value pair for every cell using the table headers as keys. 
JSON Template:
[
    {
        "content": [
            {
                "type": "table",
                "content": [
                    {
                        "key1": "value1",
                        "key2": "value2"
                    },
                    {
                        "key1": "value3",
                        "key2": "value4"
                    },
                    ... more key-value pairs ...
                ]
            },
            {
                "type": "kv",
                "content": [
                    {
                        "key1": "value1"
                    },
                    ... more key-value pairs ...
                ]
            }
        ]
    }
]
\end{lstlisting}
\end{figure}

\begin{figure}[H]
\centering
\begin{lstlisting}
Prompt of LLM-C: 
Please extract ALL key-value pairs from the following text and output only the same JSON template below. For key-value blocks, extract the pairs directly. For tables, output a key-value pair for every cell using the table headers as keys.  The input text consists of records generated from the same template, sharing the same fields and similar structures. 
JSON Template:
{
    "key1" : "value1",
    "key2" : "value2",
}
\end{lstlisting}
\end{figure}

\begin{figure}[H]
\centering
\begin{lstlisting}
Prompt of LLM-S: 
Please extract all key-value pairs from the following text and output only the same JSON template below. For key-value blocks, extract the pairs directly. For tables, output a key-value pair for every cell using the table headers as keys. The input text consists of records generated from the same template, sharing the same fields and similar structures. 
JSON Template:
[
    {
        "content": [
            {
                "type": "table",
                "content": [
                    {
                        "key1": "value1",
                        "key2": "value2"
                    },
                    {
                        "key1": "value3",
                        "key2": "value4"
                    },
                    ... more key-value pairs ...
                ]
            },
            {
                "type": "kv",
                "content": [
                    {
                        "key1": "value1"
                    },
                    ... more key-value pairs ...
                ]
            }
        ]
    }
]
\end{lstlisting}
\end{figure}